\documentclass[letterpaper,12pt]{article}

\usepackage{latexsym}
\usepackage{amssymb}

\newtheorem{defn}{Definition}
\newtheorem{eg}[defn]{Example}
\newtheorem{rem}[defn]{Remark}

\newtheorem{theorem}[defn]{Theorem}
\newtheorem{lemma}[defn]{Lemma}
\newtheorem{proposition}[defn]{Proposition}
\newtheorem{corollary}[defn]{Corollary}

\newtheorem{claim}{Claim}

\newenvironment{definition}{\defn \em }
 {\vspace{0.5 em}}

\newenvironment{example}{\eg \em }
 {\vspace{0.5 em}}

\newenvironment{remark}{\rem \em }
 {\vspace{0.5 em}}

\newenvironment{proof}
{\vspace{0.5 em} \noindent {\bf Proof.}}{\QED {\vspace{0.5 em}}}

\def\squareforqed{$\Box$}
\def\QED{\ifmmode\squareforqed\else{\unskip\nobreak\hfil
\penalty50\hskip1em\null\nobreak\hfil\squareforqed
\parfillskip = 0pt\finalhyphendemerits = 0\endgraf}\fi}

\newcommand{\ignore}[1] { }

\newcommand{\sm}[1]{\mbox{\scriptsize$#1$}}
\newcommand{\mm}[1]{\mbox{\boldmath \tiny$#1$}}

\newcommand{\A}{\mbox{$\mathcal{A}$}}

\newcommand{\cC}{\mbox{$\mathcal{C}$}}
\newcommand{\cD}{\mbox{$\mathcal{D}$}}

\newcommand{\cR}{\mbox{$\mathcal{R}$}}

\newcommand{\bv}{\mbox{\boldmath$v$}}

\newcommand{\sA}{{\mm \A}}

\newcommand{\ttU}{\mbox{$\texttt{U}$}}
\newcommand{\ttX}{\mbox{$\texttt{X}$}}

\newcommand{\sttU}{\mbox{$\texttt{\scriptsize U}$}}
\newcommand{\sttX}{\mbox{$\texttt{\scriptsize X}$}}

\newcommand{\ttLift}{\mbox{$\texttt{lift-pebble}$}}
\newcommand{\ttPlace}{\mbox{$\texttt{place-pebble}$}}
\newcommand{\ttLeft}{\mbox{$\texttt{left}$}}
\newcommand{\ttRight}{\mbox{$\texttt{right}$}}
\newcommand{\ttStay}{\mbox{$\texttt{stay}$}}
\newcommand{\ttAct}{\mbox{$\texttt{act}$}}

\newcommand{\sttLift}{\mbox{$\texttt{\scriptsize lift}$}}
\newcommand{\sttPlace}{\mbox{$\texttt{\scriptsize place}$}}

\newcommand{\sttRight}{\mbox{$\texttt{\scriptsize right}$}}
\newcommand{\sttStay}{\mbox{$\texttt{\scriptsize stay}$}}

\newcommand{\ttInc}{\mbox{$\texttt{inc}$}}
\newcommand{\ttDec}{\mbox{$\texttt{dec}$}}
\newcommand{\ttIfz}{\mbox{$\texttt{ifz}$}}

\newcommand{\sfProj}{\mbox{$\textsf{Proj}$}}
\newcommand{\sfCont}{\mbox{$\textsf{Cont}$}}

\newcommand{\sfTrue}{\mbox{$\textsf{True}$}}
\newcommand{\sfFalse}{\mbox{$\textsf{False}$}}

\newcommand{\sffqr}{\mbox{$\textsf{fqr}$}}

\newcommand{\sfEnc}{\mbox{$\textsf{Enc}$}}

\newcommand{\fD}{\mbox{$\mathfrak{D}$}}

\newcommand{\bbN}{\mbox{$\mathbb{N}$}}

\newcommand{\scriptfD}{\sm \fD}

\begin{document}

\title{On Pebble Automata for Data Languages
with Decidable Emptiness Problem\thanks{This work was done
while the author was in the Department of Computer Science
in Technion~--~Israel Institute of Technology.
It can also be found as a technical report in~\cite{Tan09TopViewTechReport}.}}

\author{
Tony Tan
\\
School of Informatics
\\
University of Edinburgh
\\
Email: \texttt{ttan@inf.ed.ac.uk}
}

\date{}
\maketitle

\begin{abstract}
In this paper we study a subclass of pebble automata (PA)
for data languages for which the emptiness problem is decidable.

Namely, we introduce the so-called {\em top view} weak PA.
Roughly speaking, top view weak PA are weak PA where
the equality test is performed only between
the data values seen by the two most recently placed pebbles.
The emptiness problem for this model is decidable.
We also show that it is robust:
alternating, nondeterministic and deterministic top view weak PA
have the same recognition power.
Moreover, this model is strong enough to accept all data languages
expressible in Linear Temporal Logic with the future-time operators,
augmented with one register freeze quantifier.
\end{abstract}

%-------------------------------------------------------------------------
\section{Introduction}
\label{s: intro}

Regular languages are clearly one of the most important concepts in computer science.
They have applications in basically all branches of computer science.
It can be argued that the following properties contributed to their success.
\begin{enumerate}
\item
Expressiveness:
In many settings regular languages are powerful enough
to capture the kinds of patterns that have to described.
\item
Decidability:
Unlike many general computational models,
the mechanisms associated with regular languages allow one to perform
automated semantic analysis.
\item
Efficiency:
The model checking problem, that is, testing whether a given string
is accepted by a given automaton can be solved in polynomial time.
\item
Closure properties:
The regular languages possess all important closure properties.
\item
Robustness:
The class of regular languages has many characterizations.
For example, various extensions like nondeterminism and alternation
do not add any expressive power.
Another characterizations include regular expressions,
monoids and monadic second-order logic.
\end{enumerate}

Moreover, similar notion of regularity has been successfully generalized
to other kind of structures,
including infinite strings and finite, as well as infinite,
ranked or unranked, trees.
Most recent applications of regular languages
(on infinite strings and finite, unranked trees, respectively)
are in model checking and XML processing.
\begin{itemize}
\item
In model checking a system is a finite state one and
properties are specified in a logic like LTL.
Satisfiability of a formula in a system is checked
on the structure that is the product of the system automaton
and an automaton corresponding to the formula.
The step from the ``real'' system to its finite state representation
usually involves many abstraction, especially with respect to
data values (variables, process numbers, etc.).
Often their range is restricted to a finite domain.
\\
Even though this approach has been successful and found its way into
large scale industrial applications, the finite abstraction have some
inherent shortcomings.
As as example, $n$ identical processes with $m$ states
each give rise to an overall model of size $m^n$.
If the number of processes is unbounded or unknown in advance,
the finite state approach fails.
Previous work has shown that even in such setting
decidability can be obtained by restricting the problem 
in various ways~\cite{AbdullaJNS04,EmersonN95}.
\item
In XML document processing,
regular concepts occur in various contexts.
First, most applications restrict the structure of the allowed documents
to conform to a certain specification (DTD or XML schema),
which can be modeled as a regular tree language.
Second, navigation (XPath) and transformation (XSLT) languages
are tightly connected to various tree automata models and
other regular description mechanism, see, for example,~\cite{Neven02}.
\\
All these approaches concentrate on the structure of the XML documents
and ignore the attribute and text values.
From a database point of view, this is not completely satisfactory,
because a {\em schema} should allow one not only to describe 
the structure of the data, but also to define restrictions on the data values
via integrity constraints such as key or inclusion constraints.
There exist a work addressing this problem~\cite{ArenasFL05},
but like in the case of model checking,
the methods rely heavily on a case-to-case analysis. 
\end{itemize}
So, in the above settings, the finite state abstraction leads to
interesting results,
but does not address all problems arising in applications.
In both cases, it would be already a big advance, if each position,
in either a string or a tree, could carry a {\em data value},
in addition to its label.

This paper is part of a broader research program
which aims at studying such extensions in a systematic way.
As any kind of operations on the infinite domain quickly leads to
undecidability of basic processing tasks (even a linear order on the domain is harmful),
we concentrate on the setting, where data values can only be tested for equality.
Furthermore, in this paper we only consider {\em finite data strings},
that is, finite strings, where each position carries a label
from a finite alphabet and a data value from an infinite domain.
Recently, there has been a significant amount of work in this direction,
see~\cite{BjorklundS07,BojanczykMSSD06,DemriL06,KaminskiF94,NevenSV04,Segoufin06}.
%and this paper aims to contribute to the progress.

Roughly speaking, there are two approaches to studying data languages:
logic and automata.
Below is a brief survey on both approaches.
For a more comprehensive survey,
we refer the reader to~\cite{Segoufin06}.
The study of data languages, which can also be viewed as languages over infinite alphabets,
% MK as we know today
starts with the introduction of {finite-memory} automata (FMA)
in~\cite{KaminskiF94},
which are also known as {\em register automata} (RA).
The study of RA was continued and extended in~\cite{NevenSV04},
in which {\em pebble automata} (PA) were also introduced.
Each of both models has its own advantages and disadvantages.
Languages accepted by FMA are closed under standard language operations:
intersection, union, concatenation, and Kleene star.
In addition, from the computational point of view,
FMA are a much easier model to handle.
Their emptiness problem is decidable,
whereas the same problem for PA is not.
However, the PA languages possess a very nice logical property:
closure under {\em all} boolean operations,
whereas FMA languages are not closed under complementation.

Later in~\cite{BojanczykMSSD06} first-order logic
for data languages was considered, and,
in particular, the so-called {\em data automata} was introduced.
It was shown that data automata define the fragment of
existential monadic second order logic for data languages in which
the first order part is restricted to two variables only.
An important feature of data automata is
that their emptiness problem is decidable, even for the infinite words,
but is at least as hard as reachability for Petri nets.
The automata themselves always work nondeterministically and
seemingly cannot be determinized, see~\cite{BjorklundS07}.
It was also shown that the satisfiability problem for
the three-variable first order logic is undecidable.

Another logical approach is via the so called
{\em linear temporal logic with
$n$ register freeze quantifier over the labels $\Sigma$},
denoted
LTL$_n^{\downarrow}(\Sigma,\texttt{X},\texttt{U})$,
see~\cite{DemriL06}.
It was shown that one way alternating $n$ register automata
accept all LTL$_n^{\downarrow}(\Sigma,\texttt{X},\texttt{U})$ languages
and the emptiness problem for one way alternating one register automata
is decidable.
Hence, the satisfiability problem for
LTL$_1^{\downarrow}(\Sigma,\texttt{X},\texttt{U})$
is decidable as well.
Adding one more register or past time operators to
LTL$_1^{\downarrow}(\Sigma,\texttt{X},\texttt{U})$
makes the satisfiability problem undecidable.

% \paragraph{\em Our paper contribution.}
In this paper we continue the study of PA,
which are finite state automata with a finite number of pebbles.
The pebbles are placed on/lifted from the input word
in the stack discipline~-- first in last out~--
and are intended to mark positions in the input word.
One pebble can only mark one position and
the most recently placed pebble serves as the head of the automaton.
The automaton moves from one state to another
depending on the current label and
the equality tests among data values in the positions
currently marked by the pebbles, as well as,
the equality tests among the positions of the pebbles.

Furthermore,
as defined in~\cite{NevenSV04},
there are two types of PA,
according to the position of the new pebble placed.
In the first type, the ordinary PA, also called {\em strong} PA,
the new pebbles are placed at the beginning of the string.
In the second type, called {\em weak} PA,
the new pebbles are placed at the position of the most recent pebble.
Obviously, two-way weak PA is just as expressive as
two-way ordinary PA.
However, it is known that
one-way nondeterministic weak PA are weaker than
one-way ordinary PA, see~\cite[Theorem~4.5.]{NevenSV04}.

We show that the emptiness problem for
one-way weak 2-pebble automata is decidable,
while the same problem for one-way weak 3-pebble automata
is undecidable.
We also introduce the so-called {\em top view} weak PA.
Roughly speaking, top view weak PA are one-way weak PA
where the equality test is performed only
between the data values seen by
the two most recently placed pebbles.
Top view weak PA are quite robust:
alternating, nondeterministic and deterministic top view weak PA
have the same recognition power.
To the best of our knowledge,
this is the first model of computation
for data language with such robustness.
It is also shown that top view weak PA
can be simulated by one-way alternating one-register RA.
Therefore, their emptiness problem is decidable.
Another interesting feature is top view weak PA
can simulate all LTL$_1^{\downarrow}(\Sigma,\ttX,\ttU)$ languages,
and the number of pebbles needed to simulate such LTL sentences
corresponds linearly to the so called {\em free quantifier rank} of the sentences,
the depth of the nesting level of the freeze operators in the sentence.

This paper is organized as follows.
In Section~\ref{s: model} we review the models of computations
for data languages considered in this paper.
Section~\ref{s: decidability weak 2-pa} and Section~\ref{s: complexity weak pa}
deals with the decidability and the complexity issues of weak PA, respectively.
In Section~\ref{s: top view weak k-pa}
we introduce top view weak PA.
We also introduce a simple extension to top view weak PA,
called unbounded top view weak PA, in which
the number of pebbles is unbounded in Section~\ref{s: unbounded} 
Finally, we end our paper with a brief observation
in Section~\ref{s: conclusion}.
This paper is augmented with appendices that contain most of the omitted details.

\section{Models of computations}
\label{s: model}

In Subsections~\ref{ss: pa} and~\ref{ss: hierarchy weak pa}
we recall the definition of weak PA from~\cite{NevenSV04},
and review the strict hierarchy of weak PA languages established in~\cite{Tan09Reach}.
In Subsection~\ref{ss: ltl}
we recall the temporal logical framework for data languages.

We will use the following notation.
We always denote by $\Sigma$ a finite alphabet of {\em labels}
and by $\fD$ an infinite set of {\em data values}.
A {\em $\Sigma$-data word}
$w =
{\sigma_1 \choose a_1}
{\sigma_2 \choose a_2}
\cdots
{\sigma_n \choose a_n}$
is a finite sequence over $\Sigma \times \fD$,
where $\sigma_i \in \Sigma$ and $a_i \in \fD$.
A {\em $\Sigma$-data language} is a set of $\Sigma$-data words.
The idea is that the alphabet $\Sigma$ is accessed directly,
while data values can only be tested for equality.

We assume that neither of $\Sigma$ and $\fD$ contain
the left-end marker $\triangleleft$ or
the right-end marker  $\triangleright$.
The input word to the automaton is of the form
$\triangleleft w \triangleright$,
where $\triangleleft$ and $\triangleright$
mark the left-end and the right-end of the input word.

We will also use the following notations.
For $w = {\sigma_1\choose a_1}\cdots {\sigma_n\choose a_n}$,
\begin{eqnarray*}
\textsf{Proj}_{\Sigma}(w) & = & \sigma_1\cdots\sigma_n
\\
\sfProj_{\scriptfD}(w) & = & a_1\cdots a_n
\\
\sfCont_{\Sigma}(w) & = & \{\sigma_1,\ldots,\sigma_n\}
\\
\sfCont_{\scriptfD}(w) & = & \{a_1,\ldots,a_n\}
\end{eqnarray*}

Finally, the symbols $\nu,\vartheta,\sigma,\ldots$,
possibly indexed, denote labels in $\Sigma$ and
the symbols $a,b,c,d,\ldots$, possibly indexed,
denote data values in $\fD$.

\subsection{Pebble automata}
\label{ss: pa}

\begin{definition}
\label{d: pa}
(See~\cite[Definition~2.3]{NevenSV04})
A {\em one-way alternating weak $k$-pebble automaton}
or, in short, $k$-PA, over $\Sigma$
is a system $\A = \langle \Sigma,Q, q_0, F, \mu, U \rangle$
whose components are defined as follows.
\begin{itemize}
\item
$Q$, $q_0 \in Q$ and $F\subseteq Q$ are
a finite set of {\em states},
the {\em initial state},
and the set of {\em final states}, respectively;
\item
$U \subseteq Q - F$ is the set of {\em universal} states; and
\item
$\mu \subseteq \cC \times \cD$ is the transition relation, where
\begin{itemize}
\item
$\cC$ is a set whose elements are of the form $(i,\sigma,V,q)$
where
$1\leq i\leq k$, $\sigma \in \Sigma$, $V \subseteq \{i+1,\ldots,k\}$
and $q\in Q$; and
\item
$\cD$ is a set whose elements are of the form $(q,\ttAct)$,
where $q \in Q$ and
$
\ttAct \in \{\ttStay,\ttRight,\ttPlace,\ttLift \}.
$
\end{itemize}
Elements of $\mu$ will be written as
$(i,\sigma,V,q) \rightarrow (p,\ttAct)$.
\end{itemize}
\end{definition}

\begin{remark}
\label{r: pebble numbering}
Note that the pebble numbering that differs from that
in~\cite{NevenSV04}.
In the above definition we adopt the pebble numbering
from~\cite{BojanczykSSS06} in which the pebbles placed
on the input word are numbered from $k$ to $i$ and
not from $1$ to $i$ as in~\cite{NevenSV04}.
The reason for this reverse numbering is that it allows us
to view the computation between placing and lifting pebble~$i$ as
a computation of an $(i-1)$-pebble automaton.

Furthermore, the automaton is no longer equipped with
the ability to compare positional equality,
in contrast with the ordinary PA introduced in~\cite{NevenSV04}.
Such ability no longer makes any difference because
the new pebbles are placed in the ``weak'' manner.
\end{remark}

Given a word
$w =
{\sigma_1 \choose a_1}
\cdots
{\sigma_n\choose a_n}
\in
(\Sigma\times\fD)^\ast$,
a {\em configuration of $\A$ on $\triangleleft w \triangleright$}
is a triple $[i,q,\theta]$, where
$i \in \{1,\ldots,k\}$, $q \in Q$, and
$\theta : \{i,i+1,\ldots,k\} \rightarrow \{ 0,1,\ldots,n,n + 1 \}$,
where $0$ and $n+1$ are positions of
the end markers $\triangleleft$ and $\triangleright$, respectively.
The function $\theta$ defines the position of the pebbles and
is called the {\em pebble assignment}.
The {\em initial} configuration is $\gamma_0 = [k,q_0,\theta_0]$,
where $\theta_0(k) = 0$ is the {\em initial} pebble assignment.
A configuration $[i,q,\theta]$ with $q \in F$ is called
an {\em accepting} configuration.

A transition $(i, \sigma ,V,p) \rightarrow \beta$
{\em applies to a configuration $[j,q,\theta]$}, if
\begin{enumerate}
\item[$(1)$]
$i = j$ and $p = q$,
\item[$(2)$]
$V = \{ l > i : a_{\theta(l)} = a_{\theta(i)} \}$, and
\item[$(3)$]
$ \sigma_{\theta(i)} = \sigma$.
\end{enumerate}

Next we define the transition relation $\vdash_{\sA}$ as follows:
$[i,q,\theta]\vdash_{\sA} [i^\prime,q^\prime,\theta^\prime]$, if
there is a transition $\alpha \rightarrow (p,\ttAct) \in \mu$
that applies to $[i,q,\theta]$ such that
$q^\prime = p$,
for all $j > i$, $\theta^\prime(j)=\theta(j)$, and
\begin{itemize}
\item[-]
if $\ttAct = \ttStay$,
then $i^\prime=i$ and $\theta^\prime(i)=\theta(i)$,
\item[-]
if $\ttAct = \ttRight$,
then $i^\prime=i$ and $\theta^\prime(i)=\theta(i)+1$,
\item[-]
if $\ttAct = \ttLift$,
then $i^\prime=i+1$,
\item[-]
if $\ttAct = \ttPlace$,
then $i^\prime=i-1$, $\theta^\prime(i-1) = \theta(i)$ and
$\theta^\prime(i)=\theta(i)$.
\end{itemize}
As usual,
we denote the reflexive transitive closure of $\vdash_{\sA}$ by
$\vdash^\ast_{\sA}$.
When the automaton $\A$ is clear from the context,
we shall omit the subscript $\A$.

The acceptance criteria is based on the notion of
{\em leads to acceptance} below.
For every configuration $\gamma=[i,q,\theta]$,
\begin{itemize}
\item
if $q \in F$, then $\gamma$ leads to acceptance;
\item
if $q \in U$, then
$\gamma$ leads to acceptance if and only if
for all configurations $\gamma'$ such that
$\gamma\vdash\gamma'$, $\gamma'$ leads to acceptance;
\item
if $q \notin F\cup U$,
then $\gamma$ leads to acceptance if and only if
there is at least one configuration $\gamma'$ such that
$\gamma \vdash \gamma'$, and $\gamma'$ leads to acceptance.
\end{itemize}
A $\Sigma$-data word $w \in (\Sigma\times\fD)^\ast$ is accepted by $\A$,
if $\gamma_0$ leads to acceptance.
The language $L(\A)$ consists of all data words accepted by $\A$.

The automaton $\A$ is {\em nondeterministic},
if the set $U=\emptyset$, and it is {\em deterministic},
if there is exactly one transition that applies
for each configuration.
It turns out that weak PA languages are quite robust.

\begin{theorem}
\label{t: equivalence}
For all $k \geq 1$,
alternating, non-deterministic and deterministic weak $k$-PA
have the same recognition power.
\end{theorem}
The proof is quite standard.
For the details of the proof, we refer the reader to
Appendix~\ref{app: s: equivalence weak pa}.

Next, we define the hierarchy of languages accepted by PA.
For $k \geq 1$,
We define the following classes of languages.
\begin{eqnarray*}
\textrm{wPA}_k & = & \{L : L \textrm{ is accepted by a weak }k\textrm{-PA}\}; \ \mbox{and}
\\
\textrm{wPA} & = &  \bigcup_{k \geq 1} \textrm{wPA}_k
\end{eqnarray*}

This example will be useful in the subsequent section.
\begin{example}
\label{e: monadic consistent}
Consider a $\Sigma$-data language $L_{\sim}$ defined as follows.
A $\Sigma$-data word
$w={\sigma_1\choose a_1} \cdots {\sigma_n\choose a_n} \in L_{\sim}$
if and only if for all $i,j=1,\ldots,n$,
if $a_i=a_j$, then $\sigma_i=\sigma_j$.
That is, $w \in L_{\sim}$ if and only if
whenever two positions in $w$ carry the same data value,
their labels are the same.

The language $L_{\sim}$ is accepted by weak $2$-PA
which works in the following manner.
Pebbles~2 iterates through all possible positions in $w$.
At each iteration, pebble~1 is placed and
scans through all the positions to the right of pebble~2,
checking whether there is a position with the same data value of pebble~2.
If there is such position, then
the labels seen by pebbles~1 and~2 are the same.
\end{example}

\subsection{Strict hierarchy of weak PA languages}
\label{ss: hierarchy weak pa}

In this section we review an example of data language
introduced in~\cite{Tan09Reach}.
It will be useful in establishing our definability results
for LTL$^{\downarrow}_{1}(\Sigma,\ttX,\ttU)$ languages.

Let $\Sigma=\{\sigma\}$ be a singleton alphabet.
For an integer $m \geq 1$,
the language $\cR^+_m$ consists of $\Sigma$-data words of the form
$$
{\sigma \choose a_0}{\sigma \choose a_1} \underbrace{\cdots}_{w_1}
{\sigma \choose a_1}{\sigma \choose a_2}
%\underbrace{\cdots}_{w_2}
%{\sigma \choose a_2}{\sigma \choose a_3}
\cdots\cdots\cdots
{\sigma \choose a_{m-2}}{\sigma \choose a_{m-1}} \underbrace{\cdots}_{w_{m-1}}
{\sigma \choose a_{m-1}}{\sigma \choose a_m}
$$
where
\begin{itemize}
\item
for each $i = 0,1,\ldots,m-1$,
$a_i\neq a_{i+1}$;
\item
for each $i = 1,\ldots,m-1$,
$a_i \not\in \sfCont_{\scriptfD}(w_i)$.
\end{itemize}
The language $\cR^+$ is defined as
$$
\cR^{+}  =  \bigcup_{m=1,2,\ldots} \cR^{+}_m.
$$

\begin{theorem}
\label{t: hierarchy weak PA}
(See~\cite[Lemma~18]{Tan09Reach}.)
For each $k=1,2,\ldots$,
\begin{enumerate}
\item
$\cR_k \in \textrm{wPA}_k$ and $\cR_{k+1} \notin \textrm{wPA}_{k}$;
\item
$\textrm{wPA}_k \subsetneq \textrm{wPA}_{k+1}$.
\end{enumerate}
\end{theorem}

\subsection{Linear temporal logic with one register freeze quantifier}
\label{ss: ltl}

In this section we recall the definition of Linear Temporal Logic (LTL)
with one register freeze quantifier~\cite{DemriL06}.
We consider only one-way temporal operators
``next'' $\ttX$ and ``until'' $\ttU$, and
do not consider their past time counterparts.

Let $\Sigma$ be a finite alphabet of labels.
Roughly, the logic LTL$_1^{\downarrow}(\Sigma,\ttX,\ttU)$
is standard LTL augmented with a register to store a data value.
Formally, the formulas are defined as follows.
\begin{itemize}
\item
Both $\sfTrue$ and $\sfFalse$ belong to
LTL$_1^{\downarrow}(\Sigma,\ttX,\ttU)$.
\item
The empty formula $\epsilon$ belongs to LTL$_1^{\downarrow}(\Sigma,\ttX,\ttU)$.
\item
For each $\sigma \in \Sigma$,
$\sigma$ is in LTL$_1^{\downarrow}(\Sigma,\ttX,\ttU)$.
\item
If $\varphi,\psi$ are in LTL$_1^{\downarrow}(\Sigma,\ttX,\ttU)$,
then so are $\neg \varphi$, $\varphi\vee\psi$ and $\varphi\wedge\psi$.
\item
$\uparrow$ is in LTL$_1^{\downarrow}(\Sigma,\ttX,\ttU)$.
\item
If $\varphi$ is in LTL$_1^{\downarrow}(\Sigma,\ttX,\ttU)$,
then so is $\ttX \varphi$.
\item
If $\varphi$ is in LTL$_1^{\downarrow}(\Sigma,\ttX,\ttU)$,
then so is $\downarrow \varphi$.
\item
If $\varphi,\psi$ are in LTL$_1^{\downarrow}(\Sigma,\ttX,\ttU)$,
then so is $\varphi\ttU\psi$.
\end{itemize}
Intuitively, the predicate $\uparrow$ is intended to mean that
the current data value is the same as the data value in the register,
while $\downarrow \varphi$ is intended to mean that
the formula $\varphi$ holds when the register contains the current data value.
This will be made precise in the definition of
the semantics of LTL$^{\downarrow}_1(\Sigma,\ttX,\ttU)$ below.

An occurrence of $\uparrow$ within the scope of
some freeze quantification $\downarrow$ is bounded by it;
otherwise, it is free.
A sentence is a formula with no free occurrence of $\uparrow$.

Next we define
the {\em freeze quantifier rank} of a sentence $\varphi$,
denoted by $\sffqr(\varphi)$.
\begin{itemize}
\item
For each $\sigma \in \Sigma$, $\sffqr(\sigma)=0$.
\item
$\sffqr(\sfTrue)=\sffqr(\sfFalse)=\sffqr(\uparrow)=0$.
\item
$\sffqr(\ttX\varphi)=\sffqr(\neg\varphi)=\sffqr(\varphi)$,
for every $\varphi$ in LTL$_1^{\downarrow}(\Sigma,\ttX,\ttU)$.
\item
$\sffqr(\varphi\vee\psi)=\sffqr(\varphi\wedge\psi)=\sffqr(\varphi\ttU\psi)=\max(\sffqr(\varphi),\sffqr(\psi))$,
for every $\varphi$ and $\psi$ in
LTL$_1^{\downarrow}(\Sigma,\ttX,\ttU)$.
\item
$\sffqr(\downarrow\varphi)=\sffqr(\varphi)+1$,
for every $\varphi$ in LTL$_1^{\downarrow}(\Sigma,\ttX,\ttU)$.
\end{itemize}

Finally,
we define the semantics of LTL$_1^{\downarrow}(\Sigma,\ttX,\ttU)$.
% MK is defined as follows.
Let $w = {\sigma_1\choose a_1} \cdots {\sigma_n\choose a_n}$
be a $\Sigma$-data word.
For a position $i = 1,\ldots,n$, a data value $a$ and
a formula $\varphi$ in LTL$_1^{\downarrow}(\Sigma,\ttX,\ttU)$,
% MK we write
$w,i \models_{a} \varphi$
% MK to mean
means that
$\varphi$ is satisfied by $w$ at position $i$
when the content of the register is $a$.
As usual,
% MK we write
$w,i \not\models_{a} \varphi$
% MK to mean that
means
$\varphi$ is not satisfied by $w$ at position $i$
when the content of the register is $a$.
The satisfaction relation is defined inductively as follows.
% MK We assume that

\begin{itemize}
\item
$w,i \models_a \epsilon$
for all $i=1,2,\ldots,n$ and $a\in\fD$.
\item
$w,i \models_a \sfTrue$
and $w,i \not\models_a \sfFalse$,
for all $i=1,2,3,\ldots$ and $a\in\fD$.
\item
$w,i \models_a\; \sigma$ if and only if
$\sigma_i =\sigma$.
\item
$w,i \models_a\; \varphi \vee \psi$ if and only if
$w,i \models_a\; \varphi$  or  $w,i \models_a \;\psi$.
\item
$w,i \models_a\; \varphi \wedge \psi$ if and only if
$w,i \models_a\; \varphi$ and  $w,i \models_a \;\psi$.
\item
$w,i \models_a\; \neg \varphi$  if and only if
$w,i \not\models_a\; \varphi$.
\item
$w,i \models_a\; \ttX \varphi$  if and only if
$1\leq i < |w|$  and
$w,i+1 \models_a\; \varphi$.
\item
$w,i \models_a\; \varphi\ttU\psi$ if and only if
there exists $j\geq i$ such that
\begin{itemize}
\item
$w,j \models_a\; \psi$ and
\item
$w,j' \models_a\; \varphi$, for all $j'=i,\ldots,j-1$.
\end{itemize}
\item
$w,i \models_a\; \downarrow\!\varphi$  if and only if
$w,i \models_{a_i}\; \varphi$
\item
$w,i \models_a\; \uparrow$  if and only if
$a = a_i$.
\end{itemize}
For a sentence $\varphi$ in LTL$_1^{\downarrow}(\Sigma,\ttX,\ttU)$,
we define the $\Sigma$-data language $L(\varphi)$ by
\begin{eqnarray*}
L(\varphi) & = &
\{w \mid w,1 \models_{a} \varphi \textrm{ for some }a \in \fD\}.
\end{eqnarray*}
Note that since $\varphi$ is a sentence,
all occurrences of $\uparrow$ in $\varphi$ are bounded.
Thus, it makes no difference which data value $a$
is used in the statement $w,1 \models_a \varphi$
of the definition of $L(\varphi)$.

\section{Decidability and undecidability of weak PA}
\label{s: decidability weak 2-pa}

In this section we will discuss the decidability issue of weak PA.
We show that the emptiness problem for weak 3-PA is undecidable,
while the same problem for weak 2-PA is decidable.
The proof of the decidability of the emptiness problem for weak 2-PA
will be the basis of the proof of the decidability of the same problem
for top view weak PA.

\begin{theorem}
\label{t: undecidable weak 3-pa}
The emptiness problem for weak $3$-PA is undecidable.
\end{theorem}
\begin{proof}
The proof is very similar to the proof of the undecidability of the
emptiness problem for weak 5-PA in~\cite{NevenSV04}.
We observe that the same proof can be easily adopted to weak 3-PA.
The details are provided below.
It uses a reduction from the Post Correspondence Problem (PCP),
which is well known to be undecidable~\cite{HopcroftU79}.
An {\em instance} of PCP is a sequence of pairs $(x_1,y_1),\ldots,(x_n,y_n)$,
where each $x_1,y_1,\ldots,x_n,y_n \in \{\alpha,\beta\}^*$.

This instance has a {\em solution} if
there exist indexes $i_1,\ldots,i_m \in \{1,\ldots,n\}$
such that $x_{i_1}\cdots x_{i_m} = y_{i_1}\cdots y_{i_m}$.
The PCP asks whether a given instance of the problem has a solution.

In the following we show how to encode a solution of an instance of PCP
into a data word which possesses properties that can be checked
by a weak 3-PA.
Let $\Sigma = \{1,\ldots,n,\alpha,\beta,\$\}$.
We denote by $x_i = \nu_{i,1}\cdots\nu_{i,l_i}$,
for each $i=1,\ldots,n$.
Each string $x_i$ is encoded as
$\sfEnc(x_i) = {\nu_{i,1}\choose a_{i,1}}\cdots {\nu_{i,l_i}\choose a_{i,l_i}}$
where $a_{i,1},\ldots,a_{i,l_i}$ are pairwise different.

The string $x_{i_1}x_{i_2}\cdots x_{i_m}$ can be encoded as
\begin{eqnarray*}
\sfEnc(x_{i_1},x_{i_2},\ldots,x_{i_m})
& = &
{i_1 \choose b_1} \sfEnc(x_{i_1})
{i_2 \choose b_2} \sfEnc(x_{i_2})
\cdots
{i_m \choose b_m} \sfEnc(x_{i_m})
\end{eqnarray*}
where all the data values that appear in it are pairwise different.
Note that even if $i_j=i_{j'}$ for some $j,j'$,
the data values that appear in $\sfEnc(x_{i_j})$
do not appear in $\sfEnc(x_{i_{j'}})$ and vice versa.
The idea is each data value is used to mark a place in the string.

Similarly, the string $y_{j_1}y_{j_2}\cdots y_{j_l}$ can be encoded as
\begin{eqnarray*}
\sfEnc(y_{j_1},y_{j_2},\ldots,y_{j_l})
& = &
{j_1 \choose c_1} \sfEnc(y_{j_1})
{j_2 \choose c_2} \sfEnc(y_{j_2})
\cdots
{j_l \choose c_l} \sfEnc(y_{j_l})
\end{eqnarray*}
where the data values that appear in it are pairwise different.

Now the data word
$$
{i_1 \choose b_1} \sfEnc(x_{i_1})
%{i_2 \choose b_2} \sfEnc(x_{i_2})
\cdots
{i_m \choose b_m} \sfEnc(x_{i_m})
{\$ \choose d}
{j_1 \choose c_1} \sfEnc(y_{j_1})
%{j_2 \choose c_2} \sfEnc(y_{j_2})
\cdots
{j_l \choose c_l} \sfEnc(y_{j_l})
$$
constitutes a solution to the instance of PCP if and only if
\begin{eqnarray}
\label{eq: pcp 1}
i_1 i_2 \cdots i_m & = & j_1 j_2 \cdots j_l
\\
\label{eq: pcp 2}
\sfProj_{\Sigma}(\sfEnc(x_{i_1})\cdots \sfEnc(x_{i_m}))
& = &
\sfProj_{\Sigma}(\sfEnc(y_{j_1})\cdots \sfEnc(y_{j_l}))
\end{eqnarray}
Now, in order to able to check such property
with weak 3-PA,
we demand the following additional criteria.
\begin{enumerate}
\item
$b_1\cdots b_m = c_1\cdots c_l$;
\item
$\sfProj_{\scriptfD}(\sfEnc(x_{i_1})\cdots \sfEnc(x_{i_m}))
 =
\sfProj_{\scriptfD}(\sfEnc(y_{j_1})\cdots \sfEnc(y_{j_l}))$
\item
For any two positions $h_1$ and $h_2$
where $h_1$ is to the left of the delimiter ${\$ \choose c}$
and $h_2$ is to the right of the delimiter ${\$ \choose c}$,
if both of them have the same data value,
then both of them are labelled with the same label.
\end{enumerate}
All the Criterias~(1)--(3) imply Equations~\ref{eq: pcp 1} and~\ref{eq: pcp 2}.

Because the data values that appears in
$\sfProj_{\scriptfD}(\sfEnc(x_{i_1}),\ldots,\sfEnc(x_{i_m}))$
are pairwise different,
all of them are checkable by three pebbles in the ``weak'' manner.
For example, to check Criteria~(1),
the automaton does the following.
\begin{itemize}
\item
Check that $b_1=c_1$.
\item
Check that for each $i=1,\ldots,m-1$,
there exists $j$ such that
$a_{i}a_{i+1}=b_{j}b_{j+1}$.
\\
It can be done by placing pebble~$3$ to read $a_i$ and pebble~$2$ to read $a_{i+1}$,
then using pebble~$3$ to search on the other side of $\$$ for
the index $j$.
\item
Finally, check that $b_m=c_l$.
\end{itemize}
Criteria~(2) can be checked similarly and
Criteria~(3) is straightforward.
The reduction is now complete and
we prove that the emptiness problem for weak $3$-PA is undecidable.
\end{proof}

Now we are going to show that
the emptiness problem for weak 2-PA is decidable.
The proof is by simulating weak $2$-PA by
one-way alternating one register automata (1-RA).
In fact, the simulation can be easily generalized to arbitrary number of pebbles.
That is, weak $k$-PA can be simulated by one-way alternating $(k-1)$-RA.
This result settles a question left open in~\cite{NevenSV04}:
Can weak PA be simulated by alternating RA?
We refer the reader to Appendix~\ref{app: s: simulation weak k-pa}
for the details of the proof.

\begin{theorem}
\label{t: simulation weak 2-pa}
For every weak $2$-PA $\A$,
there exists a one-way alternating $1$-RA $\A'$
such that $L(\A)=L(\A')$.
Moreover, the construction of $\A'$ from $\A$ is effective.
\end{theorem}

Now, by Theorem~\ref{t: simulation weak 2-pa},
we immediately obtain the decidability of weak 2-PA
because the emptiness problem for
one-way alternating $1$-RA is decidable~\cite[Theorem~4.4]{DemriL06}.

\begin{corollary}
\label{c: decidable weak 2-pa}
The emptiness problem for weak $2$-PA is decidable.
\end{corollary}

We devote the rest of this section to
the proof of Theorem~\ref{t: simulation weak 2-pa}.

Let $\A = \langle Q,q_0,\mu,F \rangle$ be a weak $2$-PA.
We assume that $\A$ is deterministic.
Furthermore, we normalize the behavior of$\A$ as follows.
\begin{itemize}
\item
Pebble~1 is lifted only after
it reads the right-end marker symbol $\triangleright$.
\item
Only pebble~2 can enter a final state and
it does so after it reads the right-end marker $\triangleright$.
\item
Immediately after pebble~2 moves right,
pebble~1 is placed.
\item
Immediately after pebble~1 is lifted,
pebble~2 moves right.
\end{itemize}

On input word $w={\sigma_1\choose d_1}\cdots {\sigma_n \choose d_n}$,
the run of $\A$ on $\triangleleft w \triangleright$
can be depicted as a tree shown in Figure~\ref{f: tree}.

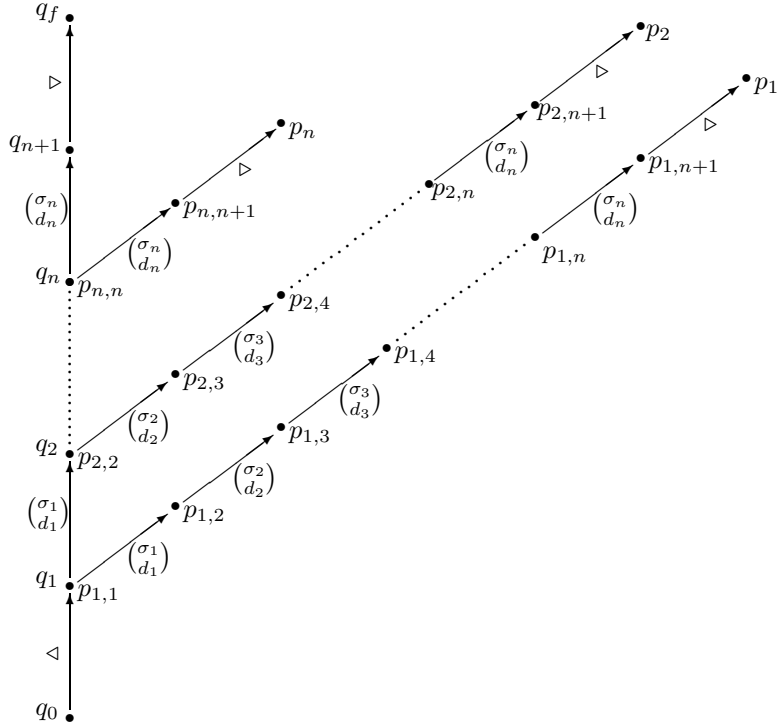
\begin{figure}[ht]

\begin{picture}(350,300)(-130,-225)

{\footnotesize

\put(-70,-210){\circle*{3}}
\put(-83,-209){$q_0$}
\put(-70,-207){\vector(0,1){44}}
\put(-79,-188){$\triangleleft$}

\put(-70,-160){\circle*{3}}
\put(-83,-159){$q_1$}
\put(-70,-157){\vector(0,1){44}}
\put(-87,-135){${\sigma_1 \choose d_1}$}

\put(-70,-110){\circle*{3}}
\put(-83,-109){$q_2$}
%\put(-70,-107){\vector(0,1){44}}
%\put(-80,-85){$a_2$}

%\put(-70,-60){\circle*{3}}
%\put(-83,-59){$q_3$}

\multiput(-70,-106)(0,3){20}{\circle*{1}}

\put(-70,-45){\circle*{3}}
\put(-83,-44){$q_{n}$}
\put(-70,-42){\vector(0,1){44}}
\put(-88,-20){${\sigma_n \choose d_n}$}

\put(-70,5){\circle*{3}}
\put(-93,6){$q_{n+1}$}
\put(-70,8){\vector(0,1){44}}
\put(-78,28){$\triangleright$}

\put(-70,55){\circle*{3}}
\put(-83,56){$q_f$}

%% first branch

\put(-67,-164){$p_{1,1}$}

\put(-67,-158.5){\vector(4,3){34}}
%\put(-62,-140){${\sigma_1 \choose d_1}$}
\put(-49,-152){${\sigma_1 \choose d_1}$}

\put(-30,-130){\circle*{3}}
\put(-27,-134){$p_{1,2}$}

\put(-27,-128.5){\vector(4,3){34}}
\put(-9,-122){${\sigma_2 \choose d_2}$}

\put(10,-100){\circle*{3}}
\put(13,-104){$p_{1,3}$}

\put(13,-98.5){\vector(4,3){34}}
\put(31,-92){${\sigma_3 \choose d_3}$}

\put(50,-70){\circle*{3}}
\put(53,-74){$p_{1,4}$}

\multiput(54,-67)(3,2.25){17}{\circle*{1}}

\put(106,-28){\circle*{3}}
\put(109,-36){$p_{1,n}$}

\put(109,-26.5){\vector(4,3){34}}
\put(127,-20){${\sigma_n \choose d_n}$}

\put(146,2){\circle*{3}}
\put(149,-2){$p_{1,n+1}$}

\put(149,3.5){\vector(4,3){34}}
\put(170,12){$\triangleright$}

\put(186,32){\circle*{3}}
\put(189,28){$p_{1}$}

%% second branch

\put(-67,-114){$p_{2,2}$}

\put(-67,-108.5){\vector(4,3){34}}
\put(-49,-102){${\sigma_2 \choose d_2}$}

\put(-30,-80){\circle*{3}}
\put(-27,-84){$p_{2,3}$}

\put(-27,-78.5){\vector(4,3){34}}
\put(-9,-72){${\sigma_3 \choose d_3}$}

\put(10,-50){\circle*{3}}
\put(13,-54){$p_{2,4}$}

\multiput(14,-47)(3,2.25){17}{\circle*{1}}

\put(66,-8){\circle*{3}}
\put(67,-12){$p_{2,n}$}

\put(68,-6.5){\vector(4,3){35}}
\put(86,0){${\sigma_n \choose d_n}$}

\put(106,22){\circle*{3}}
\put(107,18){$p_{2,n+1}$}

\put(108,23.5){\vector(4,3){35}}
\put(129,32){$\triangleright$}

\put(146,52){\circle*{3}}
\put(148,48){$p_{2}$}

%% third branch

%\put(-67,-64){$p_{3,3}$}

%\put(-67,-58.5){\vector(4,3){35}}
%\put(-56,-42){$a_3$}

%\put(-30,-30){\circle*{3}}
%\put(-27,-34){$p_{3,4}$}

%\multiput(-26,-27)(3,2.25){17}{\circle*{1}}

%\put(26,12){\circle*{3}}
%\put(27,8){$p_{3,n}$}

%\put(27,13.5){\vector(4,3){35}}
%\put(38,30){$a_{n}$}

%\put(66,42){\circle*{3}}
%\put(68,38){$p_{3,n+1}$}

%\put(68,43.5){\vector(4,3){35}}
%\put(78,60){$\triangleright$}

%\put(106,72){\circle*{3}}
%\put(107,68){$p_{3}'$}

%%  last branch

\put(-70,-45){\circle*{3}}
\put(-67,-49){$p_{n,n}$}

\put(-67,-43.5){\vector(4,3){35}}
\put(-49,-37){${\sigma_n \choose d_n}$}

\put(-30,-15){\circle*{3}}
\put(-27,-19){$p_{n,n+1}$}

\put(-27,-13.5){\vector(4,3){35}}
\put(-6,-5){$\triangleright$}

\put(10,15){\circle*{3}}
\put(13,11){$p_{n}$}

}
\end{picture}
\label{f: tree}
\caption{The tree representation of a run of $\A$ on
$w={\sigma_1\choose d_1}\cdots {\sigma_n\choose d_n}$.}
\end{figure}

The meaning of the tree is as follows.
\begin{itemize}
\item
$q_0,q_1,\ldots,q_n,q_{n+1}$ are the states of $\A$
when pebble~2 is the head pebble reading
the positions $0,1,\ldots,n,n+1$, respectively,
that is, the symbols
$\triangleleft, {\sigma_1\choose d_1},\ldots,{\sigma_n\choose d_n},\triangleright$, respectively.
\item
$q_f$ is the state of $\A$ after
pebble~2 reads the symbol $\triangleright$.
\item
For $1\leq i \leq j \leq n$, $p_{i,j}$ is the state of $\A$
when pebble~1 is the head pebble above the position $j$
while pebble~2 is above the position $i$.
\item
For $1 \leq i \leq n$,
the state $p_i$ is the state of $\A$ immediately after
pebble~1 is lifted and pebble~2 is above the position $i$.
\\
It must be noted that there is a transition
$(2,\sigma_i,\emptyset,p_i)\to (q_{i+1},\ttRight)$
applied by $\A$ that is not depicted in the figure.
\end{itemize}

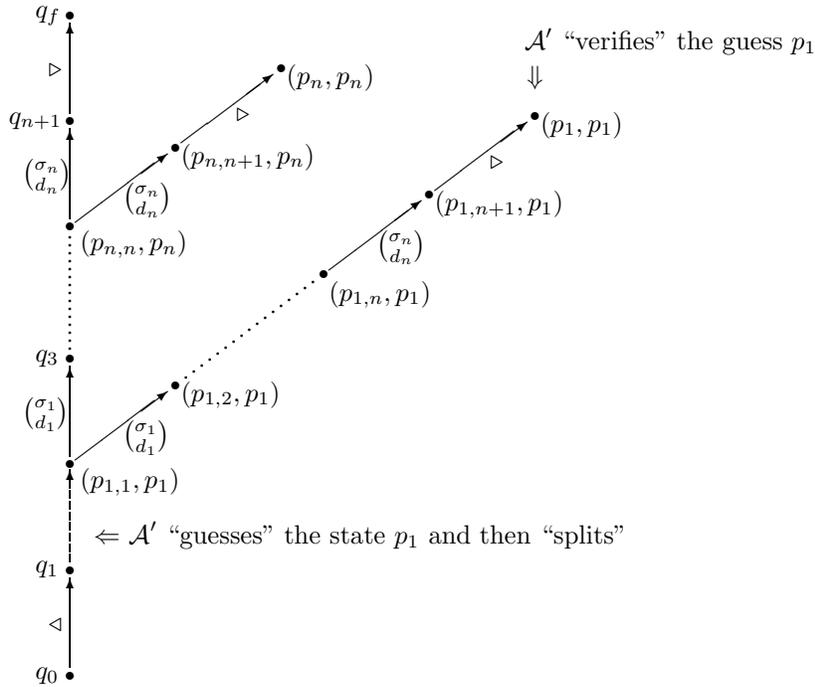
\begin{figure}[ht]

\begin{picture}(350,300)(-130,-210)

{\footnotesize

\put(-70,55){\circle*{3}}
\put(-83,54){$q_f$}

\put(-78,32){$\triangleright$}
\put(-70,18){\vector(0,1){34}}
\put(-70,15){\circle*{3}}
\put(-93,14){$q_{n+1}$}

%%%%%%%%%%%%%%%%%%%%%%%%%%%%%%%%%%%%%%%%%%%%%%%%%%%%%%%%%%
%%%%%%%%%%%%%%%%%%%%% last branch %%%%%%%%%%%%%%%%%%%%%%
%%%%%%%%%%%%%%%%%%%%%%%%%%%%%%%%%%%%%%%%%%%%%%%%%%%%%%%%%%

\put(-50,-17){${\sigma_n\choose d_n}$}
\put(-68,-23.5){\vector(4,3){35}}

\put(-30,5){\circle*{3}}
\put(-28,-1){$(p_{n,n+1},p_n)$}
\put(-7,15){$\triangleright$}
\put(-28,6.5){\vector(4,3){35}}

\put(10,35){\circle*{3}}
\put(12,29){$(p_{n},p_n)$}

%%%%%%%%%%%%%%%%%%%%%%%%%%%%%%%%%%%%%%%%%%%%%%%%%%%%%%%%%%
%%%%%%%%%%%%%%% end of  last branch %%%%%%%%%%%%%%%%%%%%%%
%%%%%%%%%%%%%%%%%%%%%%%%%%%%%%%%%%%%%%%%%%%%%%%%%%%%%%%%%%

\put(-88,-8){${\sigma_n\choose d_n}$}
\put(-70,-22){\vector(0,1){34}}
\put(-70,-25){\circle*{3}}
\put(-66,-34){$(p_{n,n},p_n)$}

%\multiput(-70,56)(0,-5){20}{\circle*{1}}
\multiput(-70,-71)(0,3){15}{\circle*{1}}

\put(-70,-75){\circle*{3}}
\put(-83,-76){$q_{3}$}

%%%%%%%%%%%%%%%%%%%%%%%%%%%%%%%%%%%%%%%%%%%%%%%%%%%%%%%%%%
%%%%%%%%%%%%%%%%%%%%% second branch %%%%%%%%%%%%%%%%%%%%%%
%%%%%%%%%%%%%%%%%%%%%%%%%%%%%%%%%%%%%%%%%%%%%%%%%%%%%%%%%%

%\put(-56,-16.5){$a_2$}
%\put(-68,-33.5){\vector(4,3){35}}

%\put(-28,-11){$(p_{2,3},p_2)$}
%\put(-30,-5){\circle*{3}}
%\put(-16,13.5){$a_3$}
%\put(-28,-3.5){\vector(4,3){35}}

%\put(12,19){$(p_{2,4},p_2)$}
%\put(10,25){\circle*{3}}
%\multiput(15,28.75)(3,2.25){18}{\circle*{1}}

%\put(72,64){$(p_{2,n},p_2)$}
%\put(70,70){\circle*{3}}
%\put(84,88.5){$a_n$}
%\put(72,71.5){\vector(4,3){35}}

%\put(112,94){$(p_{2,n+1},p_2)$}
%\put(110,100){\circle*{3}}
%\put(124,115.5){$\triangleright$}
%\put(112,101.5){\vector(4,3){35}}

%\put(152,124){$(p_{2},p_2)$}
%\put(150,130){\circle*{3}}

%%%%%%%%%%%%%%%%%%%%%%%%%%%%%%%%%%%%%%%%%%%%%%%%%%%%%%%%%%
%%%%%%%%%%%%%%% end of second branch %%%%%%%%%%%%%%%%%%%%%
%%%%%%%%%%%%%%%%%%%%%%%%%%%%%%%%%%%%%%%%%%%%%%%%%%%%%%%%%%

%\put(-80,-18){$a_2$}
%\put(-70,-32){\vector(0,1){34}}
%\put(-70,-35){\circle*{3}}
%\put(-66,-40){$(p_{2,2},p_2)$}

%\put(-70,-38){\vector(0,1){1}}
%\multiput(-70,-72)(0,4){9}{\line(0,1){3}}
%\put(-70,-75){\circle*{3}}
%\put(-83,-76){$q_{2}$}

%%%%%%%%%%%%%%%%%%%%%%%%%%%%%%%%%%%%%%%%%%%%%%%%%%%%%%%%%%
%%%%%%%%%%%%%%%%%%%%%% first branch %%%%%%%%%%%%%%%%%%%%%%
%%%%%%%%%%%%%%%%%%%%%%%%%%%%%%%%%%%%%%%%%%%%%%%%%%%%%%%%%%

\put(-50,-107){${\sigma_1\choose d_1}$}
\put(-68,-113.5){\vector(4,3){35}}

%\put(-28,-91){$(p_{1,2},p_{1})$}
%\put(-30,-85){\circle*{3}}
%\put(-16,-67.5){$a_2$}
%\put(-28,-83.5){\vector(4,3){35}}

%\put(12,-61){$(p_{1,3},p_{1})$}
%\put(10,-55){\circle*{3}}
%\put(24,-36.5){$a_3$}
%\put(12,-53.5){\vector(4,3){35}}

\put(-28,-91){$(p_{1,2},p_{1})$}
\put(-30,-85){\circle*{3}}
\multiput(-26,-82)(3,2.25){17}{\circle*{1}}

\put(28,-53){$(p_{1,n},p_{1})$}
\put(26,-43){\circle*{3}}
\put(46,-35){${\sigma_n\choose d_n}$}
\put(28,-41.5){\vector(4,3){35}}

\put(68,-19){$(p_{1,n+1},p_{1})$}
\put(66,-13){\circle*{3}}
\put(89,-3){$\triangleright$}
\put(68,-11.5){\vector(4,3){35}}

\put(108,11){$(p_{1},p_{1})$}
\put(106,17){\circle*{3}}

\put(103,29){$\Downarrow$}
\put(102,42){$\A'$ ``verifies'' the guess $p_1$}

%%%%%%%%%%%%%%%%%%%%%%%%%%%%%%%%%%%%%%%%%%%%%%%%%%%%%%%%%%
%%%%%%%%%%%%%%% end of first branch %%%%%%%%%%%%%%%%%%%%%%
%%%%%%%%%%%%%%%%%%%%%%%%%%%%%%%%%%%%%%%%%%%%%%%%%%%%%%%%%%
\put(-88,-98){${\sigma_1\choose d_1}$}
\put(-70,-112){\vector(0,1){34}}
\put(-70,-115){\circle*{3}}
\put(-66,-124){$(p_{1,1},p_1)$}

\put(-70,-118){\vector(0,1){1}}
\multiput(-70,-152)(0,4){9}{\line(0,1){3}}
\put(-61,-145){$\Leftarrow$ $\A'$ ``guesses'' the state $p_1$ and then ``splits''}
\put(-70,-155){\circle*{3}}
\put(-83,-156){$q_{1}$}

\put(-78,-178){$\triangleleft$}
\put(-70,-192){\vector(0,1){34}}
\put(-70,-195){\circle*{3}}
\put(-83,-196){$q_0$}

}

\end{picture}
\label{f: tree run alternating FMA}
\caption{The corresponding run of $\A'$
to the one in Figure~\ref{f: tree}.}
\end{figure}

Now the simulation of $\A$ by a one-way alternating 1-RA $\A'$
becomes straightforward by transforming the tree in Figure~\ref{f: tree}
into a tree depicting the computation of $\A'$
on the same word $w$.

Roughly, the automaton $\A'$ is defined as follows.
\begin{itemize}
\item
The states of $\A'$ are elements of $Q \cup (Q\times Q)$\footnote{Actually $\A'$ needs some other
auxiliary states. However, for the intuitive explanation here the set $Q \cup (Q\times Q)$ suffices.
We refer the reader to Appendix~\ref{app: s: simulation weak k-pa} for the details.};
\item
the initial state is $q_0$; and
\item
the set of final states is $F \cup \{(p,p):p \in Q\}$.
\end{itemize}
For each placement of pebble~$1$ on position $i$,
the automaton performs the following ``Guess--Split--Verify'' procedure
which consists of the following steps.
\begin{enumerate}
\item
From the state $q_i$, $\A'$ ``guesses'' the state
in which pebble~1 is eventually lifted, i.e. the state $p_i$,
and stores it in its internal state.
\\
That is, $\A'$ enters into the state $(q_i,p_i)$.
\item
$\A'$ ``splits'' its computation (conjunctively) into two branches.
\begin{itemize}
\item
In one branch, assuming that the guess $p_i$ is correct,
$\A'$ moves right and enters into the state $q_{i+1}$,
simulating the transition $(2,\emptyset,p_i)\to(q_{i+1},\ttRight)$.
After this, it recursively performs the Guess--Split--Verify procedure
for the next placement of pebble~1 on position~$(i+1)$.
\item
In the other branch $\A'$ stores the data value $d_i$ in its register
and simulates the run of pebble~1 on ${\sigma_i \choose d_i}\cdots {\sigma_n\choose d_n}$
to ``verify'' that the guess $p_i$ is correct.
\\
That is, $\A'$ accepts only if it ends in the state $(p_i,p_i)$.
\end{itemize}
\end{enumerate}
Figure~\ref{f: tree run alternating FMA} shows
the corresponding run of $\A'$ on the same word.

\section{Complexity of weak 2-PA}
\label{s: complexity weak pa}

In this subsection we are going to determine the
time complexity of three specific problems related to weak 2-PA.
\begin{description}
\item[Emptiness problem.]
The emptiness problem for weak $2$-PA.
That is, given a weak 2-PA $\A$, is $L(\A)=\emptyset$?
\item[Labelling problem.]
Given a weak $2$-PA $\A$ over the labels $\Sigma$
and a sequence of data values $d_1\cdots d_n\in \fD^n$,
is there a sequence of labels $\sigma_1\cdots\sigma_n \in \Sigma^n$
such that ${\sigma_1 \choose d_1}\cdots{\sigma_n \choose d_n} \in L(\A)$?
\item[Data value membership problem.]
Given a weak $2$-PA $\A$ over the labels $\Sigma$
and a sequence of finite labels $\sigma_1\cdots \sigma_n \in \Sigma^n$,
is there a sequence of data values $d_1\cdots d_n \in \fD^n$
such that ${\sigma_1 \choose d_1}\cdots{\sigma_n \choose d_n} \in L(\A)$?
\end{description}
The emptiness problem, as we have seen in the previous section, is decidable.
The labelling and data value membership problem are definitely decidable.
To solve the labelling problem,
one simply iterates all possible sequence $\sigma_1\cdots\sigma_n \in \Sigma^n$
and runs $\A$ to check whether ${\sigma_1 \choose d_1}\cdots{\sigma_n\choose d_n} \in L(\A)$.
Such straightforward algorithm requires $O(|\Sigma|^n\cdot n^2)$ computational steps.
Similarly, to solve the data value membership problem,
one can iterate all possible sequence of data values $d_1\cdots d_n$
and run $\A$ to check whether ${\sigma_1 \choose d_1}\cdots{\sigma_n\choose d_n} \in L(\A)$.
Since the word is of length $n$,
one simply needs to consider up to $n$ different data values.
Such algorithm takes $O(n^n\cdot n^2)$ computational steps.

We are going to show that
the emptiness problem is not primitive recursive,
while both the labelling and data value membership problems
are NP-complete.

We start the proof with a few simple examples
of languages accepted by weak 2-PA.
Though simple, they are very crucial
in determining the complexity of the emptiness problem for weak 2-PA.

\begin{example}
\label{e: incremental weak 2-pa}
Let $\Sigma = \{\alpha,\beta\}$.
We define the $\Sigma$-data language $L_{inc}$
which consists of the data words of the following form:
$$
\underbrace{{\alpha \choose a_1} \cdots {\alpha \choose a_m}}_{w_1}
\underbrace{{\beta \choose b_1} \cdots {\beta \choose b_n}}_{w_2},
$$
where
\begin{itemize}
\item
the data values $a_1,\ldots,a_m$ are pairwise different;
\item
the data values $b_1,\ldots,b_n$ are pairwise different;
\item
$\sfProj_{\Sigma}(w_1)=\alpha^m$;
\item
$\sfProj_{\Sigma}(w_2)=\beta^n$;
\item
$\sfCont_{\scriptfD}(w_1) \subseteq \sfCont_{\scriptfD}(w_2)$.
\end{itemize}
All these conditions can be checked by weak 2-PA.
The intention of data words in $L_{inc}$ is to represent the inequality $m \leq n$.
\end{example}

\begin{example}
\label{e: incremental +l weak 2-pa}
Let $\Sigma = \{\alpha,\beta\}$.
For a fixed $l \geq 0$,
we define the language $L_{inc,+1}$
which consists of the data words of the following form:
$$
\underbrace{{\alpha \choose a_1} \cdots {\alpha \choose a_m}}_{w_1}
\underbrace{{\beta \choose b_1} \cdots {\beta \choose b_n}}_{w_2}
$$
where
\begin{itemize}
\item
the data values $a_1,\ldots,a_m$ are pairwise different;
\item
the data values $b_1,\ldots,b_n$ are pairwise different;
\item
$\sfProj_{\Sigma}(w_1)=\alpha^m$;
\item
$\sfProj_{\Sigma}(w_2)=\beta^{n}$;
\item
For each $a_i \in \sfCont_{\scriptfD}(w_1)$,
$a_i \neq b_1$.
\item
$\{a_1,\ldots,a_m\} \subseteq \{b_2,\ldots,b_n\}$.
\end{itemize}
Again, all these conditions can be checked by weak 2-PA.
The intention of data words in $L_{inc,+1}$ is to represent the inequality $m+1 \leq n$.
\end{example}

\begin{example}
\label{e: incremental -l weak 2-pa}
Let $\Sigma = \{\alpha,\beta\}$.
For a fixed $l \geq 0$,
we define the language $L_{inc,-1}$
which consists of the data words of the following form:
$$
\underbrace{{\alpha \choose a_1} \cdots {\alpha \choose a_{m}}}_{w_1}
\underbrace{{\beta \choose b_1} \cdots {\beta \choose b_n}}_{w_2}
$$
where
\begin{itemize}
\item
the data values $a_1,\ldots,a_m$ are pairwise different;
\item
the data values $b_1,\ldots,b_n$ are pairwise different;
\item
$\sfProj_{\Sigma}(w_1)=\alpha^m$;
\item
$\sfProj_{\Sigma}(w_2)=\beta^n$;
\item
The symbol $a_1 \notin \{b_{1},\ldots,b_n\}$;
\item
For each $i=2,\ldots,m$,
$a_i \in \{b_{1},\ldots,b_n\}$.
\end{itemize}
Again, all these conditions can be checked by weak 2-PA.
The intention of data words in $L_{inc,-1}$ is to represent the inequality $m-1 \leq n$.
\end{example}

\begin{theorem}
\label{t: compexity weak 2-pa}
The emptiness problem for weak $2$-PA is not primitive recursive.
\end{theorem}
\begin{proof}
The proof is by simulation of incrementing counter automata.
It follows closely the proof of similar lower bound for
one-way alternating 1-RA~\cite[Theorem~2.9]{DemriL06}.
It is known that the emptiness problem for incrementing counter automata
is decidable~\cite[Theorem~6]{Mayr03}, but not primitive recursive~\cite{Schnoebelen02}.
We refer the reader to Appendix~\ref{app: s: counter automata} for
the formal definition of incrementing counter automata.

In short, an incrementing $l$-counter automaton over $\Sigma$ is
an automaton with $l$ counters, operates on words over $\Sigma$,
and the value in each counter is allowed to erroneously increase,
hence, the name {\em incrementing}.

A configuration is a tuple $(q,\sigma,\bv)$
where $q$ is a state, $\sigma$ is the current symbol read
and $\bv:\{1,\ldots,l\}\to\bbN$,
where $\bv(i)$ denotes the value stored in counter $i$.

Now a configuration $(q,\bv)$ can be encoded as a
$(Q \cup \Sigma \cup \{c_1,\ldots,c_l\})$-data word as follows.
$$
{q\choose d_1}{\sigma\choose d_2}
{c_1 \choose a_{1,1}}\cdots{c_1\choose a_{1,\bv(1)}}
\cdots
{c_l \choose a_{l,1}}\cdots{c_l\choose a_{l,\bv(l)}}.
$$
where the symbols $a_{1,1},\ldots,a_{l,\bv(l)}$ are pairwise different.
The labels $c_1,\ldots,c_l$ are used as pointers
that the current data value is part of the encoding of
the counters $\bv(1),\ldots,\bv(l)$, respectively.

Since the automaton allows for erroneous increment of values in each counter,
we can check the validity of the application of each transition,
like in Examples~\ref{e: incremental weak 2-pa},~\ref{e: incremental +l weak 2-pa}
and~\ref{e: incremental -l weak 2-pa}.
\end{proof}

Now we are going to show the NP-completeness of
the labelling problem.
It is by a reduction from graph 3-colorability problem.

Given an undirected graph $G=(V,E)$,
let $V = \{1,\ldots,n\}$ and $E = \{(i_1,j_1),\ldots,(i_m,j_m)\}$.
We can take $i_1j_1 \cdots i_mj_m$ as the sequence of data values.
Then, we construct a weak $2$-PA $\A$ over the alphabet $\Sigma = \{\vartheta_R,\vartheta_G,\vartheta_B\}$
that accepts data words of even length in which the following hold.
\begin{itemize}
\item
For all odd position $x$,
the label on position $x$ is different from
the label on position $x+1$.
\item
For every two positions $x$ and $y$,
if they have the same data value,
then they have the same label.
\end{itemize}
Thus, the graph $G$ is 3-colorable if and only if
there exists $\sigma_1\cdots\sigma_{2m} \in \{\vartheta_R,\vartheta_G,\vartheta_B\}^*$ such that
$$
{\sigma_1\choose i_1} {\sigma_2\choose j_1}\cdots
{\sigma_{2m-1}\choose i_m} {\sigma_{2m}\choose j_m}
\in L(\A),
$$
and the NP-completeness of the labeling problem follows.

The NP-completeness of data value membership problem can established in a similar spirit.
The reduction is from the following variant
of graph 3-colorability,
called 3-colorability with constraint.
Given a graph $G=(V,E)$ and three integers $n_r$, $n_g$, $n_b$ in {\em unary} form,
can the graph $G$ be colored with the colors $R$, $G$ and $B$ such that
the numbers of vertices colored with $R$, $G$ and $B$
are $n_r$, $n_g$ and $n_b$, respectively?

The polynomial time reduction to data value membership problem is as follows.
Let $V = \{1,\ldots,n\}$
and $E = \{(i_1,j_1),\ldots,(i_m,j_m)\}$.

We define $\Sigma=\{\vartheta_R,\vartheta_G,\vartheta_B,\nu_1,\ldots,\nu_n\}$.
We take
$$
\nu_{i_1}\nu_{j_1}\cdots \nu_{i_m}\nu_{j_m}
\underbrace{\vartheta_R\cdots \vartheta_R}_{n_r\textrm{ \scriptsize times}}
\underbrace{\vartheta_G \cdots \vartheta_G}_{n_g\textrm{ \scriptsize times}}
\underbrace{\vartheta_B\cdots \vartheta_B}_{n_b\textrm{ \scriptsize times}}
$$
as the sequence of finite labels.

Then, we construct a weak $2$-PA over $\Sigma$
that accepts data words of the form
$$
{\nu_{i_1} \choose c_1}{\nu_{j_1} \choose d_1}
\cdots
{\nu_{i_m} \choose c_m}{\nu_{j_m} \choose d_m}
{\vartheta_R \choose a_1}\cdots {\vartheta_R \choose a_{n_r}}
{\vartheta_G \choose a_{1}'}\cdots {\vartheta_G \choose a_{n_g}'}
{\vartheta_B \choose a_{1}''}\cdots {\vartheta_B \choose a_{n_b}''}
$$
where
\begin{itemize}
\item
$\nu_{i_1},\nu_{j_1},\ldots,\nu_{i_m},\nu_{j_m} \in \{\nu_1,\ldots,\nu_n\}$;
\item
in the sub-word
${\nu_{i_1} \choose c_1} {\nu_{j_1} \choose d_1} \cdots {\nu_{i_m} \choose c_m} {\nu_{j_m} \choose d_m}$,
every two positions with the same labels have the same data value
, see Example~\ref{e: monadic consistent};
\item
the data values $a_1,\ldots,a_{n_r},a_1',\ldots,a_{n_g}',a_1'',\ldots,a_{n_b}''$ are pairwise different;
\item
For each $i=1,\ldots,m$,
the data values $c_i,d_i$ appear among $a_1,\ldots,a_{n_r}$,
$a_1',\ldots,a_{n_g}'$, $a_1'',\ldots,a_{n_b}''$
such that the following holds:
\begin{itemize}
\item
if $c_i$ appears among $a_1,\ldots,a_{n_r}$,
then $d_i$ appears among $a_{1}',\ldots,a_{n_g}'$ or $a_1'',\ldots,a_{n_b}''$;
\item
if $c_i$ appears among $a_{1}',\ldots,a_{n_g}'$,
then $d_i$ appears either among $a_1,\ldots,a_{n_r}$ or $a_{1}'',\ldots,a_{n_b}''$; and
\item
if $c_i$ appears among $a_{1}'',\ldots,a_{n_b}''$,
then $d_i$ appears among $a_{1},\ldots,a_{n_r}$ or $a_1',\ldots,a_{n_g}'$.
\end{itemize}
\end{itemize}
Note that we can store the integers $r$, $g$, $b$ and $m$
in the internal states $\A$,
thus, enable $\A$ to ``count'' up to $n_r,n_g,n_b$ and $m$.
We have each state for the numbers $1,\ldots,n_r$, $1,\ldots,n_g$, $1,\ldots,n_b$
and $1,\ldots,m$.
Furthermore, the unary form of $n_r$, $n_g$ and $n_b$ is crucial here
to ensure that the number of the states of $\A$
is still polynomial in the length of the input.

Now the graph $G$ is $3$-colorable with constraint
if and only if there exits
$c_1 d_1 \cdots c_m d_m a_1\cdots a_{n_r} a_1'\cdots a_{n_g}' a_1''\cdots a_{n_b}''$ such that
$$
{\nu_{i_1}\choose c_{1}}{\nu_{j_1}\choose d_{1}}
\cdots
{\nu_{i_m}\choose c_{m}}{\nu_{j_m}\choose d_{m}}
{\vartheta_R\choose a_{1}} \cdots {\vartheta_R\choose a_{n_r}}
{\vartheta_G\choose a_{1}'} \cdots {\vartheta_G\choose a_{n_g}'}
{\vartheta_B\choose a_{1}''} \cdots {\vartheta_B\choose a_{n_b}''}
$$
is accepted by $\A$, and
the NP-completeness of data value membership problem follows.

\section{Top view weak $k$-PA}
\label{s: top view weak k-pa}

In this section we are going to restrict the definition
of weak $k$-PA so that its emptiness problem becomes decidable.
Roughly speaking, {\em top view} weak PA are weak PA
where the equality test is performed only
between the data values seen by the last and the second last placed pebbles.
That is, if pebble~$i$ is the head pebble,
then it can only compare the data value it reads
with the data value read by pebble~$(i+1)$.
It is not allowed to compare its data value with
those read by pebble~$i+2,\ldots,k$.

Formally, the transitions of top view weak $k$-PA
$\A = \langle Q,q_0,\mu,F\rangle$
are of the form
$$
(i,\sigma,V,q) \to (q',\ttAct)
$$
where $V$ is either $\emptyset$ or $\{i+1\}$.

The definition of top view weak $k$-PA
is defined by setting
$$
V =
\left\{
\begin{array}{ll}
\emptyset, & \textrm{if } a_{\theta(i+1)} \neq a_{\theta(i)}
\\
\{i+1\}, & \textrm{if } a_{\theta(i+1)} = a_{\theta(i)}
\end{array}
\right.
$$
in the definition of transition relation in Subsection~\ref{ss: pa}.
Note that top view weak $2$-PA are just the same
as weak $2$-PA.
We can also define the alternating version
of top view weak $k$-PA.
However, just like in the case of weak $k$-PA,
alternating, nondeterministic and deterministic top view weak $k$-PA
have the same recognition power.

\begin{theorem}
\label{t: top view weak k-pa in alternating 1ra}
For every top view weak $k$-PA $\A$,
there is a one-way alternating 1-RA $\A'$
such that $L(\A')=L(\A)$.
Moreover, the construction of $\A'$ is effective.
\end{theorem}
\begin{proof}
The proof is a straightforward generalization of the proof of Theorem~\ref{t: simulation weak 2-pa}.
Each placement of a pebble is simulated by
``Guess--Split--Verify'' procedure.
Since each pebble~$i$ can only compare its data value
with the one seen by pebble~$i+1$,
$\A'$ does not need to store the data values seen by pebble~$i+2,\ldots,k$.
It only need to store the data value seen by pebble~$i+1$,
thus, one register suffices.
\end{proof}

Following Theorem~\ref{t: top view weak k-pa in alternating 1ra},
we immediately obtain the decidability
of the emptiness problem for top view weak $k$-PA.
\begin{corollary}
\label{c: top view k-pa is decidable}
The emptiness problem for top view weak $k$-PA is decidable.
\end{corollary}

\begin{remark}
\label{r: tight bound}
Since the emptiness problem for ordinary 2-PA and
for weak 3-PA is already undecidable
(See Theorem~\ref{t: undecidable weak 3-pa} and~\cite[Theorem 4]{KaminskiT08}),
it seems that top view weak PA is a tight boundary
of a subclass of PA languages for which the emptiness problem is decidable.
\end{remark}

\begin{theorem}
\label{t: LTL1 is in top view weak k-pa}
For every sentence $\psi \in \textrm{LTL}^{\downarrow}(\Sigma,\ttX,\ttU)$,
there exists a top view weak $k$-PA $\A_{\psi}$,
where $k = \sffqr(\psi)+1$,
such that $L(\A_{\psi})=L(\psi)$.
\end{theorem}
\begin{proof}
Let $\psi$ be an $\textrm{LTL}^{\downarrow}(\Sigma,\ttX,\ttU)$ sentence.
We construct an alternating top view weak $k$-PA $\A_{\psi}$,
where $k=\sffqr(\psi)+1$ such that
given a data word $w$, the automaton $\A_{\psi}$ checks whether $w,1\models \psi$.
$\A_{\psi}$ accepts if it is so. Otherwise, it rejects.

Intuitively, the computation of $w,1\models\psi$ is done recursively as follows.
The automaton $\A_{\psi}$ ``consists of'' the automata $\A_{\varphi}$
for all sub-formula $\varphi$ of $\psi$, including $\A_{\epsilon}$
to represent the empty formula $\epsilon$.
\begin{itemize}
\item
The automaton $\A_{\epsilon}$ accepts every data words.
\item
If $\psi = \sigma \varphi$,
then check whether the current label is $\sigma$.
If it is not, then $\A$ rejects immediately.
Otherwise, $\A_{\psi}$ proceeds to run $\A_{\varphi}$.
\item
If $\psi = \varphi \vee \varphi'$,
then $\A_{\psi}$ nondeterministically chooses one of $\A_{\varphi}$ or $\A_{\varphi'}$
and proceeds to run one of them.
\item
If $\psi = \varphi \wedge \varphi'$,
then $\A_{\psi}$ splits its computation (by conjunctive branching) into two
and proceed to run both of $\A_{\varphi}$ and $\A_{\varphi'}$.
\item
If $\psi = \ttX \varphi$,
then $\A_{\psi}$ moves to the right one step.
If it reads the right-end marker, then the automaton rejects immediately.
Otherwise, it proceeds to run $\A_{\varphi}$.
\item
If $\psi = \uparrow \varphi$,
then $\A_{\psi}$ checks whether the data value seen by its head pebble
is the same as the one seen by the second last placed pebble.
If it is not the same, then it rejects immediately.
Otherwise, it proceeds to run $\A_{\varphi}$.
\item
If $\psi = \downarrow\varphi$,
then $\A_{\psi}$ places a new pebble and
proceeds to run $\A_{\varphi}$.
\item
If $\psi = \varphi \ttU \varphi'$,
then $\A_{\psi}$ it runs $\A_{\varphi' \vee (\varphi \wedge \sttX(\varphi\sttU\varphi'))}$.
\item
If $\psi = \neg \varphi$,
then $\A_{\psi}$ runs $\A_{\varphi}$.
If $\A_{\varphi}$ accepts, then $\A_{\psi}$ rejects. Otherwise, $\A_{\psi}$ accepts.
\end{itemize}
Note that since $\sffqr(\varphi)=k$,
on each computation path then the automaton $\A_{\psi}$ only needs to place the pebble $k$ times,
thus, $\A_{\psi}$ requires only $k+1$.
It is a straight forward induction to show that $L(\A_{\psi})=L(\psi)$.
\end{proof}

Our next results deals with
the expressive power of LTL$^{\downarrow}_1(\Sigma,\ttX,\ttU)$
based on the freeze quantifier rank.
It is an analog of the classical hierarchy of first order logic
based on the ordinary quantifier rank.
We start by defining an LTL$^{\downarrow}_1(\Sigma,\ttX,\ttU)$ sentence
for the language $\cR^+_m$ defined in Subsection~\ref{ss: hierarchy weak pa}.

\begin{lemma}
\label{l: R_k for LTL}
For each $k=1,2,3,\ldots$,
there exists a sentence
$\psi_{k}$ in LTL$^{\downarrow}_1(\Sigma,\ttX,\ttU)$
such that $L(\psi_{k})=\cR^+_{k}$ and
\begin{itemize}
\item
$\sffqr(\psi_1)=1$; and
\item
$\sffqr(\psi_{k})=k-1$, when $k\geq 2$.
\end{itemize}
\end{lemma}
\begin{proof}
First, we define a formula $\varphi_k$
such that $\sffqr(\varphi_k)=k-1$ and
for every data word $w = {\sigma\choose d_1}\cdots {\sigma\choose d_n}$,
for every $i=1,\ldots,n$,
\begin{eqnarray}
\label{eq: phi_k for R_k^+}
w,i \models_{d_i} \varphi_k
& \textrm{if and only if} &
{\sigma\choose d_i}\cdots {\sigma\choose d_n} \in \cR^+_{k}.
\end{eqnarray}
We construct $\varphi_k$ inductively as follows.
\begin{itemize}
\item
$\varphi_1 := \ttX(\neg \uparrow) \wedge \neg(\ttX\sfTrue)$.
\item
For each $k=1,2,3,\ldots$,
\begin{eqnarray*}
\varphi_{k+1} & := &
\ttX(\neg \uparrow)
\wedge
\ttX
\Big(
\downarrow
\ttX
\Big(
(\neg\uparrow) \ttU (\uparrow \wedge \varphi_k)
\Big)
\Big)
\end{eqnarray*}
\end{itemize}
Note that since $\sffqr(\varphi_1)=0$,
then for each $k=1,2,\ldots$,
$\sffqr(\varphi_k)=k-1$.

Assuming first that $\varphi_k$ satisfies
the property in Equation~\ref{eq: phi_k for R_k^+},
the desired sentence $\psi_k$ is defined as follows.
\begin{itemize}
\item
$\psi_1 := \downarrow \big(\ttX(\neg \uparrow) \wedge \neg(\ttX\sfTrue)\big)$.
\item
For each $k=2,3,\ldots$,
\begin{eqnarray*}
\psi_{k} & := &
\downarrow(\ttX(\neg \uparrow))
\wedge
\ttX
\Big(
\downarrow
\ttX
\Big(
(\neg\uparrow) \ttU (\uparrow \wedge \varphi_{k-1})
\Big)
\Big)
\end{eqnarray*}
\end{itemize}
Since $\sffqr(\varphi_{k-1})=k-2$,
then $\sffqr(\psi_k)=k-1$.

Now we want to show that the formula $\varphi_k$
satisfies Equation~\ref{eq: phi_k for R_k^+}.
The proof is by induction on $k$.
The base case, $k=1$, is trivial.
Suppose, for the induction hypothesis,
the formula $\varphi_k$ satisfies Equation~\ref{eq: phi_k for R_k^+}.

The induction step is as follows.
Let $w = {\sigma \choose d_1}\cdots{\sigma \choose d_n}$.
We have the following chain of application of the semantics of LTL.
\begin{eqnarray*}
w,i & \models_{d_i} & \varphi_{k+1}
\\
& \Updownarrow &
\\
w,i & \models_{d_i} & \ttX(\neg\uparrow) \wedge
\ttX\big( \downarrow \ttX((\neg\uparrow)\ttU (\uparrow \wedge \varphi_k) )\big)
\\
& \Updownarrow &
\\
w,i & \models_{d_i} & \ttX(\neg \uparrow) \qquad\textrm{and}\qquad
\begin{array}{rcl}
w,i  & \models_{d_i} &  \ttX\big( \downarrow \ttX((\neg\uparrow)\ttU (\uparrow \wedge \varphi_k) )\big)
\end{array}
\end{eqnarray*}
For the first part,
we have
\begin{eqnarray}
\label{eq: derivation 1}
w,i  \models_{d_i}  \ttX(\neg \uparrow)
& \textrm{if and only if} &
d_i \neq d_{i+1}
\end{eqnarray}
Now we evaluate the second part.
\begin{eqnarray}
\nonumber
w,i  & \models_{d_i} &
\ttX\big( \downarrow \ttX((\neg\uparrow)\ttU (\uparrow \wedge \varphi_k) )\big)
\qquad\qquad\qquad\qquad\qquad\qquad\qquad
\\
\nonumber
& \Updownarrow &
\\
\nonumber
w,i+1  & \models_{d_i} &
\big( \downarrow \ttX((\neg\uparrow)\ttU (\uparrow \wedge \varphi_k) )\big)
\qquad\qquad\qquad\qquad\qquad\qquad\qquad
\\
\nonumber
& \Updownarrow &
\\
\nonumber
w,i+1  & \models_{d_{i+1}} &
\ttX((\neg\uparrow)\ttU (\uparrow \wedge \varphi_k) )
\qquad\qquad\qquad\qquad\qquad\qquad\qquad
\\
\nonumber
& \Updownarrow &
\\
\label{eq: derivation 2}
w,i+2  & \models_{d_{i+1}} &
(\neg\uparrow)\ttU (\uparrow \wedge \varphi_k)
\qquad\qquad\qquad\qquad\qquad\qquad\qquad
\end{eqnarray}
Equation~\ref{eq: derivation 2} holds if and only if
there exists $j$ such that $i+2\leq j$ and
\begin{enumerate}
\item
$w,j \models_{d_{i+1}} \uparrow \wedge \varphi_k$,
\item
$w,j' \models_{d_{i+1}} \neg \uparrow$, for each $j'=i+1,\ldots,j-1$.
\end{enumerate}
By the semantics of LTL and the induction hypothesis,
Clause~1 holds if and only if $d_j=d_{i+1}$ and
${\sigma\choose d_j}\cdots{\sigma\choose d_n} \in \cR_k^+$.
The meaning of Clause~2 is $d_j' \neq d_{i+1}$, for each $j'=i+1,\ldots,j-1$.
Both clauses, together with Equation~\ref{eq: derivation 1},
means that ${\sigma\choose d_i}\cdots {\sigma\choose d_n} \in \cR_{k+1}^+$.
This completes the induction hypothesis.
\end{proof}

\begin{lemma}
\label{l: R_k+1 not in fqr k-1}
For each $k=1,2,\ldots$,
the language $\cR^+_{k+1}$ is not expressible
by a sentence in LTL$^{\downarrow}_1(\Sigma,\ttX,\ttU)$
of freeze quantifier rank $k-1$.
\end{lemma}
\begin{proof}
By Theorem~\ref{t: hierarchy weak PA},
$\cR^+_{k+1}$ is not accepted by weak $k$-PA,
thus, it is also not accepted by top-view $k$-PA.
Then, by Theorem~\ref{t: LTL1 is in top view weak k-pa},
$\cR^+_{k+1}$ is not expressible by LTL$^{\downarrow}(\Sigma,\ttX,\ttU)$
sentence of freeze quantifier rank $k-1$.
\end{proof}

Combining both Lemmas~\ref{l: R_k for LTL} and~\ref{l: R_k+1 not in fqr k-1},
we obtain the following strict hierarchy of LTL$^{\downarrow}(\Sigma,\ttX,\ttU)$
based on its freeze quantifier rank.
\begin{theorem}
\label{t: fqr k+1 > fqr k}
For each $k=1,2,\ldots$,
the class of sentences in LTL$^{\downarrow}(\Sigma,\ttX,\ttU)$
of freeze quantifier rank $k+1$
is strictly more expressive than those of freeze quantifier rank $k$.
\end{theorem}

\section{Top view weak $k$-PA}
\label{s: top view weak k-pa}

In this section we are going to define {\em top view} weak PA.
Roughly speaking, top view weak PA are weak PA
where the equality test is performed only
between the data values seen by the last and the second last placed pebbles.
That is, if pebble~$i$ is the head pebble,
then it can only compare the data value it reads
with the data value read by pebble~$(i+1)$.
It is not allowed to compare its data value with
those read by pebble~$(i+2),(i+3),\ldots,k$.

Formally, top view weak $k$-PA is a tuple
$\A = \langle \Sigma, Q,q_0,\mu,F\rangle$
where $Q,q_0,F$ are as usual and
$\mu$ consists of transitions of the form:
$(i,\sigma,V,q) \to (q',\ttAct)$,
where $V$ is either $\emptyset$ or $\{i+1\}$.

The criteria for the application of transitions of top view weak $k$-PA
is defined by setting
$$
V =
\left\{
\begin{array}{ll}
\emptyset, & \textrm{if } a_{\theta(i+1)} \neq a_{\theta(i)}
\\
\{i+1\}, & \textrm{if } a_{\theta(i+1)} = a_{\theta(i)}
\end{array}
\right.
$$
in the definition of transition relation in Subsection~\ref{ss: pa}.
Note that top view weak $2$-PA and weak $2$-PA are the same.

\begin{remark}
\label{r: alt, det top view}
We can also define the alternating version
of top view weak $k$-PA.
However, just like in the case of weak $k$-PA,
alternating, nondeterministic and deterministic top view weak $k$-PA
have the same recognition power.
Furthermore, by using the same proof presented in Section~\ref{s: complexity weak pa},
it is straightforward to show that
the emptiness problem, the labelling problem, and
the data value membership problem have the same complexity lower bound
for top view weak $k$-PA, for each $k=2,3,\ldots$.
\end{remark}

The following theorem is a stronger version of Theorem~\ref{t: simulation weak 2-pa}.

\begin{theorem}
\label{t: top view weak k-pa in alternating 1ra}
For every top view weak $k$-PA $\A$,
there is a one-way alternating 1-RA $\A'$
such that $L(\A')=L(\A)$.
Moreover, the construction of $\A'$ is effective.
\end{theorem}
\begin{proof}
The proof is a straightforward generalization of the proof of Theorem~\ref{t: simulation weak 2-pa}.
Each placement of a pebble is simulated by
``Guess--Split--Verify'' procedure.
Since each pebble~$i$ can only compare its data value
with the one seen by pebble~$(i+1)$,
$\A'$ does not need to store the data values seen by pebbles~$(i+2),\ldots,k$.
It only needs to store the data value seen by pebble~$(i+1)$,
thus, one register is sufficient for the simulation.
\end{proof}

Following Theorem~\ref{t: top view weak k-pa in alternating 1ra},
we immediately obtain the decidability
of the emptiness problem for top view weak $k$-PA.
\begin{corollary}
\label{c: top view k-pa is decidable}
The emptiness problem for top view weak $k$-PA is decidable.
\end{corollary}

Since the emptiness problem for ordinary 2-PA (See~\cite[Theorem~4]{KaminskiT08}) and
for weak 3-PA is already undecidable,
it seems that top view weak PA is a tight boundary
of a subclass of PA languages for which the emptiness problem is decidable.

\begin{remark}
\label{r: LTL1 is in top view weak k-pa}
In~\cite{Tan09Reach} it is shown that
for every sentence $\psi \in \textrm{LTL}_1^{\downarrow}(\Sigma,\ttX,\ttU)$,
there exists a weak $k$-PA $\A_{\psi}$,
where $k = \sffqr(\psi)+1$,
such that $L(\A_{\psi})=L(\psi)$.
We remark that the proof actually shows that
the automaton $\A_{\psi}$ is top view weak $k$-PA.
Thus, it shows that the class of top view weak $k$-PA languages
contains the languages definable by $\textrm{LTL}_1^{\downarrow}(\Sigma,\ttX,\ttU)$.
\end{remark}

\section{Top view weak PA with unbounded number of pebbles}
\label{s: unbounded}

This section contains our quick observation on top view weak PA.
We note that the finiteness of the number of pebbles for top view weak PA
is not necessary.
In fact, we can just define top view weak PA with unbounded number of pebbles,
which we call top view weak unbounded PA.

We elaborate on it in the following paragraphs.
Let $\A = \langle\Sigma,Q,q_0,\mu,F\rangle$ be top view weak unbounded PA.
The pebbles are numbered with the numbers $1,2,3,\ldots$.
The automaton $\A$ starts the computation with only pebble~1 on the input word.
The transitions are of the form:
$(\sigma,\chi,q) \to (p,\ttAct)$,
where $\chi \in \{0,1\}$ and $\sigma,q,p,\ttAct$
are as in the ordinary weak PA.

Let $w = {\sigma_1 \choose a_1} \cdots {\sigma_n\choose a_n}$ be an input word.
A {\em configuration of $\A$ on $\triangleleft w \triangleright$}
is a triple $[i,q,\theta]$, where
$i \in \bbN$, $q \in Q$, and
$\theta : \bbN \rightarrow \{ 0,1,\ldots,n,n + 1 \}$.
The {\em initial} configuration is $[1,q_0,\theta_0]$,
where $\theta_0(1)=0$.
The accepting configurations
are defined similarly as in ordinary weak PA.

A transition $(\sigma ,\chi,p) \rightarrow \beta$
{\em applies to a configuration $[i,q,\theta]$}, if
\begin{enumerate}
\item[$(1)$]
$p = q$, and $ \sigma_{\theta(i)} = \sigma$,
\item[$(2)$]
$\chi = 1$ if  $a_{\theta(i-1)} = a_{\theta(i)}$, and
$\chi = 0$ if  $a_{\theta(i-1)} \neq a_{\theta(i)}$,
\end{enumerate}
%The transition relation $\vdash$ and the acceptance criteria
%can be defined in a similar manner as in Section~\ref{ss: pa}.

Similarly, the transition relation $\vdash$ can be defined
as follows:
$[i,q,\theta]\vdash_{\sA} [i^\prime,q^\prime,\theta^\prime]$, if
there is a transition $\alpha \rightarrow (p,\ttAct) \in \mu$
that applies to $[i,q,\theta]$ such that
$q^\prime = p$,
for all $j < i$, $\theta^\prime(j)=\theta(j)$, and
\begin{itemize}
\item[-]
if $\ttAct = \ttRight$,
then $i^\prime=i$ and $\theta^\prime(i)=\theta(i)+1$,
\item[-]
if $\ttAct = \ttLift$,
then $i^\prime=i-1$
\item[-]
if $\ttAct = \ttPlace$,
then $i^\prime=i+1$, $\theta^\prime(i+1) = \theta^\prime(i)=\theta(i)$.
\end{itemize}
The acceptance criteria can be defined similarly.

It is straightforward to show that
$1$-way deterministic 1-RA can be simulated by top view weak unbounded PA.
Each time the register automaton change the content of the register,
the top view weak unbounded PA places a new pebble.

Furthermore, top view weak unbounded PA can be simulated by 1-way alternating 1-RA.
Each time a pebble is placed, the register automaton
performs ``Guess--Split--Verify'' procedure described in Section~\ref{s: decidability weak 2-pa}.
Thus, the emptiness problem for top view unbounded weak PA is still decidable.

\section{Concluding remark}
\label{s: conclusion}

In this paper we study pebble automata for data languages.
In particular, we establish a fragment of PA languages
for which the emptiness problem is decidable,
the so called top view weak PA.
As shown in this paper, top view weak PA inherit some nice properties
mentioned in Section~\ref{s: intro}.
\begin{enumerate}
\item 
Expressiveness: 
Top view weak PA strictly contain the languages expressible by 
LTL$^{\downarrow}_{1}(\Sigma,\ttX,\ttU)$.
\item
Decidability:
The emptiness problem is decidable.
\item
Efficiency:
The model checking problem, that is, 
testing whether a given string of length $n$ is accepted by a specific 
deterministic top view weak $k$-PA can be solved in $O(n^k)$ computation time.
\item
Closure properties:
Top view weak $k$-PA languages are closed under all boolean operations.
\item
Robustness:
Alternation and nondeterminism do not add expressive power
to top view weak $k$-PA languages.
\end{enumerate}
There are still lots of work to be done.
In order to be applicable in program verification and XML settings,
the model should work be infinite strings and unranked trees, respectively.
Thus, the question remains whether it is possible to 
extend top view weak PA to the settings of infinite strings and unranked trees,
while still preserving the five properties mentioned above.

\paragraph*{\em \bf Acknowledgment.}
The author would like to thank Michael Kaminski for his
invaluable directions and guidance related to this paper and
for pointing out the notion of unbounded pebble automata.

%-------------------------------------------------------------------------
%\nocite{ex1,ex2}
%\bibliographystyle{latex8}
%\bibliography{latex8}

\appendix

\section{Counter Automata}
\label{app: s: counter automata}

A {\em Minsky $k$-counter automata} (CA),
with $\epsilon$-transitions and zero testing,
is a tuple $\A = \langle \Sigma, Q, q_0, \delta, F \rangle$,
where
\begin{itemize}
\item
$\Sigma$ is a finite alphabet;
\item
$Q$ is a finite set of states;
\item
$q_0$ is the initial state;
\item
$\delta\subseteq Q \times (\Sigma \cup \{\epsilon\})\times L \times Q $
is a transition relation over the instruction set
$L = \{\ttInc,\ttDec,\ttIfz\}\times\{1,\ldots,k\}$;
\item
$F\subseteq Q$ is the set of accepting set,
such that $q' \notin F$ whenever $(q,\epsilon,l,q')\in\delta$.
\end{itemize}
Given a word $w = \sigma_1\cdots\sigma_n \in \Sigma^*$,
a {\em configuration} of $\A$ is a triple $[i,q,\bv]$
where $0\leq i \leq n $, $q \in Q$ and
a counter valuation $\bv : \{1,\ldots,k\}\to \bbN$,
where $\bbN$ is the set of natural number $\{1,2,3,\ldots\}$.

The {\em initial} configuration is $[0,q_0,\bv_0]$
where $\bv_0(j)=0$ for each $j=1,\ldots,k$.
The {\em run} of $\A$ on $w$
is a sequence
$[0,q_0,\bv_0],[1,q_1,\bv_1],\ldots,[n,q_n,\bv_n]$
where
for each $i=0,\ldots,n-1$,
there exists a transition $(q_i,\sigma_{i+1},l,q_{i+1}) \in \delta$
and
\begin{itemize}
\item
if $l=(\ttInc,j)$ for some $j=1\ldots,k$,
then $\bv_{i+1}(j)=\bv_{i}(j)+1$
and for all other $j'\neq j$,
$\bv_{i+1}(j')=\bv_{i}(j')$.
\item
if $l=(\ttDec,j)$ for some $j=1\ldots,k$,
then $\bv{i}(j)>0$ and $\bv_{i+1}(j)=\bv_{i}(j)-1$
and for all other $j'\neq j$,
$\bv_{i+1}(j')=\bv_{i}(j')$.
\item
if $l=(\ttIfz,j)$ for some $j=1\ldots,k$,
then $\bv_{i}(j)=0$
and $\bv_{i+1}=\bv_{i}$.
\end{itemize}
The word $w$ is accepted by $\A$ if $q_n \in F$.
As usual, we denote by $L(\A)$ the set of
all words over $\Sigma$ accepted by $\A$.

We say that the automaton $\A$ is {\em incrementing}
if its counters may erroneously increase at any time.
More precisely,
The {\em run} of an incrementing $\A$ on $w$
is a sequence of configurations
$[0,q_0,\bv_0],[1,q_1,\bv_1],\ldots,[n,q_n,\bv_n]$
where
for each $i=0,\ldots,n-1$,
there exists a transition $(q_i,\sigma_{i+1},l,q_{i+1}) \in \delta$
and
\begin{itemize}
\item
if $l=(\ttInc,j)$ for some $j=1\ldots,k$,
then $\bv_{i+1}(j)\geq \bv_{i}(j)+1$
and for all other $j'\neq j$,
$\bv_{i+1}(j')\geq \bv_{i}(j')$.
\item
if $l=(\ttDec,j)$ for some $j=1\ldots,k$,
then $\bv{i}(j)>0$ and $\bv_{i+1}(j)\geq\bv_{i}(j)-1$
and for all other $j'\neq j$,
$\bv_{i+1}(j')\geq\bv_{i}(j')$.
\item
if $l=(\ttIfz,j)$ for some $j=1\ldots,k$,
then $\bv_{i}(j)=0$
and for all $j'=1,\ldots,k$,
$\bv_{i+1}(j')\geq\bv_{i}(j')$.
\end{itemize}

\begin{theorem}
\label{app: t: incrementing A}
\cite[Theorem~2.9]{DemriL06}
(See also~\cite[Theorem~6]{Mayr03} and~\cite{Schnoebelen02})
The nonemptiness problem for incrementing counter automata
is decidable, but not primitive recursive.
\end{theorem}

\section{Register automata}
\label{app: s: ra}

We are only going to sketch roughly the definition of register automata.
Readers interested in its more formal treatment can consult~\cite{DemriL06,KaminskiF94}.
In essence, $k$ register automaton, or, shortly $k$-RA,
is a finite state automaton equipped with a header to scan the input
and $k$ registers, numbered from $1$ to $k$.
Each register can store exactly one data value from $\fD$.
The automaton is {\em two-way} if the header can move to the left or to the right.
It is {\em alternating} if
it is allowed to branch into a finite number of parallel computations.

More formally, a two-way alternating $k$-RA over the label $\Sigma$
is a tuple $\A = \langle \Sigma, Q, q_0, u_0, \mu, F \rangle$ where
\begin{itemize}
\item
$Q_0$, $q_0 \in Q$ and $F\subseteq Q$
are the finite state of states, the initial state
and the set of final states, respectively.
\item
$u_0 = a_1\cdots a_k$ is the initial content of the registers.
\item
$\mu$ is a set of transitions of the following form.
\begin{itemize}
\item[{\em i})]
$(q,\sigma) \to  q'$ where $a\in \{\triangleleft,\triangleright\}$ and $q,q'\in Q$.
\\
That is, if the automaton $\A$ is in state $q$ and
the header is currently reading either of the symbols $\triangleleft,\triangleright$,
then the automaton can enter the state $q'$.
\item[{\em ii})]
$(q,\sigma,V)  \to  q'$ where $\sigma\in\Sigma$, $V \subseteq \{1,\ldots,k\}$ and $q,q'\in Q$.
\\
That is, if the automaton $\A$ is in state $q$ and
the header is currently reading a position labeled with $\sigma$ and
$V$ is the set of all registers containing the current data value,
then the automaton can enter the state $q'$.
\item[{\em iii})]
$q \to (q',I)$ where  $I \subseteq \{1,\ldots,k\}$ and $q,q'\in Q$.
\\
That is,
if the automaton $\A$ is in state $q$,
then the automaton can enter the state $q'$
and store the current data value into the registers whose indices belong to $I$.
\item[{\em iv})]
$q \to (q_1 \wedge \cdots \wedge q_i)$
and
$q \to (q_1 \vee \cdots \vee q_i)$ where $i \geq 1$ and $q,q'\in Q$.
\\
That is, if the automaton $\A$ is in state $q$,
then it can decide to perform
conjunctive or disjunctive branching into the states $q_1,\ldots,q_i$.
\item[{\em v})]
$q \to (q',\ttAct)$ where $\ttAct \in \{\ttLeft,\ttRight\}$ and $q,q'\in Q$.
\\
That is, if the automaton $\A$ is in state $q$,
then it can enter the state $q'$ and
move to the next or the previous word position.
\end{itemize}
\end{itemize}
A register automaton is called {\em non deterministic}
if the branchings of state (in item~({\em iv})) are all disjunctive.
It is called {\em one-way} if
the header is not allowed to move to the previous word position.

A {\em configuration} $\gamma = [j,q,b_1\cdots b_k]$
of the automaton $\A$ consists of
the current position of the header in the input word $j$,
the state of the automaton $q$ and the content of the registers $b_1\cdots b_k$.
The configuration $\gamma$ is called {\em accepting}
if the state is a final state in $F$.

From each configuration $\gamma$,
the automaton performs legitimate computation according to the transition relation
and enters another configuration $\gamma'$.
If the transition is branching,
then it can split into several configurations $\gamma_1',\ldots,\gamma_m'$.

Similarly, we can define the notion of {\em leads to acceptance}
for a configuration $\gamma$ as follows.
\begin{itemize}
\item
Every accepting configuration leads to acceptance.
\item
If $\gamma'$ is the configuration obtained from $\gamma$
by applying a non-branching transition,
then $\gamma$ leads to acceptance if and only if $\gamma'$ leads to acceptance.
\item
If $\gamma_1',\ldots,\gamma_m'$ are the configurations obtained from $\gamma$
by applying a disjunctive branching transition,
then $\gamma$ leads to acceptance if and only if at least one of $\gamma_1',\ldots,\gamma_m'$ leads to acceptance.
\item
If $\gamma_1',\ldots,\gamma_m'$ are the configurations obtained from $\gamma$
by applying a conjunctive branching transition,
then $\gamma$ leads to acceptance if and only if all of $\gamma_1',\ldots,\gamma_m'$ lead to acceptance.
\end{itemize}
An input word $w$ is accepted by $\A$ if
the initial configuration leads to acceptance.
As usual, $L(\A)$ denotes the language accepted by $\A$.

\section{Generalization of Theorem~\ref{t: simulation weak 2-pa}}
\label{app: s: simulation weak k-pa}

Let $\A = \langle Q, q_0, \mu, F \rangle$ be a  weak $k$-PA.
We will show how to construct one-way alternating $(k-1)$-RA $\A'$.
For our convenience, we assume that $\A$ is deterministic.
We also assume that $\A$ behaves as follows.
\begin{itemize}
\item
For every configuration $\gamma$ of $\A$,
there exists a transition in $\mu$ that applies to it.
\item
Only pebble~$k$ can enter a final state and
it does so only after it reads the right-end marker $\triangleright$.
\item
For every $i=2,\ldots,k$,
immediately after pebble~$i$ moves right,
pebble~$i-1$ is placed.
\item
For every $i=1,\ldots,k-1$,
pebble~$i$ is lifted only when
it reaches the right-end marker $\triangleright$.
\item
For every $i=1,\ldots,k-1$,
immediately after pebble~$i$ is lifted,
pebble~$(i+1)$ moves right.
\end{itemize}
See Subsection~\ref{app: ss: nondet to deterministic}
on how this normalization can be done.

We also assume that
the set of states $Q$ is partitioned into $Q_1\cup\cdots \cup Q_k$
where $Q_i\cap Q_j$ whenever $i\neq j$
and $Q_i$ is the set of states when pebble~$i$ is in control.

The automaton $\A' = \langle Q', q_0', u_0, \mu', F' \rangle$
is defined as follows.
\begin{itemize}
\item
The set of states is
$Q' = Q \cup Q^2 \cup Q^3 \cup \widetilde{Q}\cup \widetilde{Q\times Q}$,
where $\widetilde{Q}=\{\tilde{q} \mid q\in Q\}$ and
$\widetilde{Q\times Q}= \{\widetilde{(q,p)} \mid p,q \in Q\}$.
\item
The initial state is $q_0'=q_0 \in Q_k$.
\item
The initial assignment is $\#^{k-1}$.
\item
The set of final states is $F' = F \cup \{(q,q) : q\in Q\}$.
\end{itemize}
For our convenience, we number the registers of $\A'$
from $2$ to $k$, not from $1$ to $(k-1)$.
The set of transitions $\mu'$ consists of the following.
\begin{itemize}
\item
For $i=k-1,\ldots,1$, we have the following transitions.
\begin{enumerate}
\item
For each transition $(i,\sigma,V,q)\to (q',\ttRight)\in \mu$,
there are transitions $((q,p),\sigma,V) \to (q',p) \in \mu'$ for all $p \in Q_{i+1}$.
\item
For each transition $(i,P,V,q)\to(q',\ttPlace) \in \mu$,
there are the following transitions in $\mu'$.
For every $p \in Q_{i+1}$,
\begin{eqnarray*}
(q,p) & \to & \widetilde{(q,p)},\{i\}
\\
\widetilde{(q,p)} & \to & \bigvee_{p_j \in Q_i} (q,p_j,p)
\\
(q,p_j,p) & \to & (p_j,p) \wedge (q',p_j) \ \ \textrm{for every }p_j \in Q_i
\end{eqnarray*}
\end{enumerate}
\item
For $i=k$, there are the following transitions in $\mu'$.
\begin{enumerate}
\item
For each transition $(k,\sigma,\emptyset,q)\to(q',\ttRight) \in \mu$,
there is a transition $(q,\sigma,\emptyset)\to(q') \in \mu'$.
\item
For each transition $(k,\sigma,\emptyset,q)\to(q',\ttPlace) \in \mu$,
there are the following transitions in $\mu'$.
\begin{eqnarray*}
q & \to & \tilde{q},\{k\}
\\
\tilde{q} & \to & \bigvee_{p_j \in Q_k} (q,p_j)
\\
(q,p_j) & \to & p_j \wedge (q',p_j) \ \ \textrm{for every }p_j \in Q_k
\end{eqnarray*}
\end{enumerate}
\end{itemize}

We can show the following proposition by straightforward induction on $i$.
\begin{proposition}
\label{p: i-run lead to acceptance}
Let $w={\sigma_1\choose a_1}\cdots {\sigma_n\choose a_n}$ be a $\Sigma$-data word.
For $i=1,\ldots,k-1$,
there exists an $i$-run $[i,q_1,\theta_1]\vdash^*[i,q_2,\theta_2]$ of $\A$ on $w$
and $[i,q_2,\theta_2]\vdash [i+1,q_3,\theta_3]$
if and only if
the configuration $[\theta(i),(q_1,q_3),u_2\cdots u_k]$ of $\A'$ on $w$
leads to acceptance,
where $u_j=a_{\theta(j)}$, for $j=i+1,\ldots,k$.
\end{proposition}

Then, by the definition of $\mu'$,
we can easily deduce the following.
For each $\ell = 1,\ldots,n$,
$$
[k,q_1,\theta_1]\vdash_{\sA} [k-1,q_2,\theta_2] \vdash^*_{\sA}
[k-1,q_3,\theta_3] \vdash_{\sA} [k,q_4,\theta_4] \vdash_{\sA} [k,q_5,\theta_5]
$$
is a $k$-run of $\A$ on $w$,
where
$\theta_1(k)=\theta_2(k)=\theta_3(k)=\theta_4(k)=\ell$ and $\theta_5(k)=\ell+1$ and
$\theta_2(k-1)=\ell$, $\theta_3(k-1)=n+1$
if and only if
\begin{eqnarray*}
\lbrack\ell,q_1,\#^{k-2}a_{\ell-1}\rbrack & \vdash & \lbrack\ell,\tilde{q}_1,\#^{k-2}a_{\ell}\rbrack
\\
\lbrack\ell,\tilde{q}_1,\#^{k-2}a_{\ell}\rbrack & \vdash & \lbrack\ell,(q_1,q_4),\#^{k-2}a_{\ell}\rbrack
\\
\lbrack\ell,(q_1,q_4),\#^{k-2}a_{\ell}\rbrack & \vdash & \lbrack\ell,q_4,\#^{k-2}a_{\ell}\rbrack
\\
\lbrack\ell,(q_1,q_4),\#^{k-2}a_{\ell}\rbrack & \vdash & \lbrack\ell,(q_2,q_4),\#^{k-2}a_{\ell}\rbrack
\\
\lbrack\ell,q_4,\#^{k-2}a_{\ell}\rbrack & \vdash & \lbrack\ell+1,q_5,\#^{k-2}a_{\ell}\rbrack
\end{eqnarray*}
and the configuration $[\ell,(q_2,q_4),\#^{k-2}a_{\ell}]$ leads to acceptance.

Now, the equivalence between $L(\A)$ and $L(\A')$ follows immediately.

\section{Equivalence between alternating and deterministic one-way weak $k$-PA}
\label{app: s: equivalence weak pa}

For every one-way alternating weak $k$-PA,
we will construct its equivalent one-way deterministic weak $k$-PA.
This is done in two steps.
\begin{enumerate}
\item
First, we transform the one-way alternating weak $k$-PA into
its equivalent one-way nondeterministic weak $k$-PA.
\item
Then, we transform the one-way nondeterministic weak $k$-PA
into its equivalent one-way deterministic weak $k$-PA.
\end{enumerate}
We present step~2 first.

\subsection{From nondeterministic to deterministic}
\label{app: ss: nondet to deterministic}

We start with the simple case.
We will show how to determinize nondeterministic weak $2$-PA.
The idea can be generalized to arbitrary number of pebbles.

Let $\A = \langle Q, q_0, F, \mu \rangle$
be a nondeterministic weak $2$-PA.
We start by normalizing the behavior of $\A$ as follows.
\begin{itemize}
\item[N1.]
For every configuration $\gamma$ of $\A$,
there exists a transition in $\mu$ that applies to it.
\item[N2.]
Only pebble~$2$ can enter a final state and
it does so only after it reads the right-end marker $\triangleright$.
\item[N3.]
Immediately after pebble~$2$ moves right,
pebble~$1$ is placed.
\item[N4.]
Pebble~$1$ is lifted only when
it reaches the right-end marker $\triangleright$.
\end{itemize}
Such normalization can be done by adding some extra states to $\A$.
This normalization N4 is especially important, as it implies that
nondeterminism on pebble~1 is now limited
only to deciding which state to enter.
There is no nondeterminism in choosing which action to take,
i.e. either to lift pebble~1 or to keep on moving right.

Next, we note that immediately after pebble~$1$ is lifted,
there can be two choices of actions for pebble~2:
\begin{itemize}
\item[-]
to place pebble~$1$ again; or
\item[-]
moves pebble~$2$ to the right.
\end{itemize}
The following fifth normalization is supposed to handle this situation:
\begin{itemize}
\item[N5.]
Immediately after pebble~$1$ is lifted,
pebble~$2$ moves right.
\end{itemize}
In other words, while pebble~$2$ is
reading a specific position, pebble~$1$ makes exactly one pass, from
the position of pebble~$2$ to the right end of the input,
instead of making several rounds of passes by placing pebble~$1$
again immediately after it is lifted.
Since there are only finitely many states,
there can only be finitely many passes.
So, the normalization N4 can be achieved by simultaneously simulating all
possible passes in one pass.

With the normalization N1--N5,
there is no nondeterminism in choosing which action
to take for pebble~2.
The same as for pebble~1,
the nondeterminism for pebble~2 is now limited
only in deciding which states to take.
This is summed up in the following remark.
\begin{remark}
\label{r: effect normalization}
For each $i=1,2$,
if $(i,P,V,p)\to(q_1,\ttAct_1)$ and $(i,P,V,p)\to(q_2,\ttAct_2)$,
then $\ttAct_1=\ttAct_2$.
\end{remark}

Now that the nondeterminism is reduced to
deciding which state to enter,
the determinization of $\A$ becomes straightforward.
Similar to the classical proof of the equivalence between
nondeterministic and deterministic finite state automata,
we can take the power set of the states of $\A$
to deterministically simulate $\A$.

Now the normalization steps N1--N5 can be performed similarly
for weak $k$-PA $\A$.
\begin{itemize}
\item[N1$'$.]
For every configuration $\gamma$ of $\A$,
there exists a transition in $\mu$ that applies to it.
\item[N2$'$.]
Only pebble~$k$ can enter a final state and
it does so only after it reads the right-end marker $\triangleright$.
\item[N3$'$.]
For each $i=2,\ldots,k$,
immediately after pebble~$i$ moves right,
pebble~$(i-1)$ is placed.
\item[N4$'$.]
For each $i=1,\ldots,k-1$,
pebble~$i$ is lifted only when
it reaches the right-end marker $\triangleright$.
\item[N5$'$.]
For each $i=1,\ldots,k-1$,
Immediately after pebble~$i$ is lifted,
pebble~$i+1$ moves right.
\end{itemize}
Such normalization results in reducing the nondeterminism to
deciding which state to enter.
Then, the determinization of $\A$ can be done just like in the classical case
as in the case of weak 2-PA described above.

The following are the details of the determinization of $\A = \langle \Sigma,Q,q_0,F\rangle$.

Let $P,V \subseteq \{1,\ldots,k\}$. For a subset $S \subseteq Q$, we
define the following:
\begin{eqnarray*}
E_{i,P,V}(S) & = & \{q : \exists q' \in S \textrm{ such that }
(i,P,V,q')\to(q,\ttAct) \in \mu \};
\\
A_{i,P,V}(S) & = & \ttAct \textrm{ where } (i,P,V,q) \to (q',\ttAct)
\textrm{ for some }q \in S \textrm{ and }q'\in Q.
\end{eqnarray*}
By Remark~\ref{r: effect normalization} above, $A_{i,P,V}(S)$ is
well defined. Furthermore, $E_{i,P,V}$ is monotone, that is, if $S
\subseteq S'$, then $E_{i,P,V}(S)\subseteq E_{i,P,V}(S')$.

Now we present the construction of a deterministic, weak $k$-pebble
automaton $\A'$ equivalent to $\A$. Let $\A' = \langle Q', q_0', F',
\mu' \rangle$ where
\begin{itemize}
\item[-]
$Q' = 2^{Q}$;
\item[-]
$q_0' = \{q_0\}$;
\item[-]
$F' = \{S \subseteq Q : S \cap F \neq \emptyset\}$;
\item[-]
$\mu'$ consists of the following transitions. For each $i \in
\{1,\ldots,k\}$, $P,V \subseteq \{i+1,\ldots,k\}$ and $S \subseteq
Q$,
$$
(i,P,V,S) \to (E_{i,P,V}(S),A_{i,P,V}(S)) \in \mu'.
$$
\end{itemize}
%To prove that $L(\A')=L(\A)$,
%we need an auxiliary notion.
Recall that an {\em $i$-run} is a run from an $i$-configuration to
an $i$-configuration in which pebble~$i+1$ is never lifted.
%We define the following.
%For an $i$-configuration $\gamma$,
%$$
%C(\gamma) = \{\gamma' : \gamma \vdash^{*}_{\sA} \gamma' \textrm{ is an }i\textrm{-run}\}.
%$$
We are going to use the following Claim~\ref{c: proof equivalence
weak PA} to prove that $L(\A')=L(\A)$.

\begin{claim}
\label{c: proof equivalence weak PA} Let $i=1,\ldots,k$. Let $w \in
\Sigma^\ast$ and $\theta_1,\theta_2$ be pebble assignments on $w$.
\begin{enumerate}
\item
For every set of states $S_1, S_2 \subseteq Q$, if
$[i,S_1,\theta_1]\vdash^\ast [i,S_2,\theta_2]$ is an $i$-run of
$\A'$, then
$$
\bigcup_{q_1\in S_1} \{\gamma : [i,q_1,\theta_1] \vdash^{*}_{\sA}
\gamma \textrm{ is an }i\textrm{-run}\}
 = \{[i,q_2,\theta_2]: q_2 \in S_2\}.
$$
\item
For every state $q_1, q_2 \subseteq Q$, if
$[i,q_1,\theta_1]\vdash^\ast [i,q_2,\theta_2]$ is an $i$-run of
$\A$, then for all $S_1 \subseteq Q$ such that $q_1 \in S_1$, there
exists $S_2 \subseteq Q$ such that $q_2 \in S_2$ and
$[i,S_1,\theta_1]\vdash^\ast [i,S_2,\theta_2]$ is an $i$-run of
$\A'$.
\end{enumerate}
\end{claim}
\begin{proof}
Let $w,\theta_1,\theta_2$ be as above. The proof of the claim is by
induction on $i$. The base case, $i=1$, is the same as the standard
finite state automaton, thus, omitted.

For the induction hypothesis, we assume that the claim is true for
the case of $i-1$. To proceed with the induction step, we prove the
claim for the case $i$.

We start by proving $(1)$. Let $S_1, S_2\subseteq Q$. Assume that
$$
[i,S_1,\theta_1]\vdash^\ast [i,S_2,\theta_2].
$$
By our normalization of $\A$, the resulting automaton $\A'$ from our
construction is also normalized as the automaton $\A$. Thus, an
$i$-run of $\A'$ is a repeated sequence of transitions relations of
the form:
$$
[i,R_1,\lambda_1]\vdash [i-1,R_2,\lambda_2]\vdash^\ast [i-1, R_3,
\lambda_3] \vdash [i,R_4,\lambda_4]\vdash [i,R_5,\lambda_5],
$$
where
\begin{itemize}
\item[{\em a)}]
some transition $(i,P_1,V_1,R_1)\to (R_2,\ttPlace)\in \mu'$ is
applied to obtain the transition relation $[i,R_1,\lambda_1]\vdash
[i-1,R_2,\lambda_2]$,
\item[{\em b)}]
the transition relation $[i-1,R_2,\lambda_2]\vdash^\ast [i-1, R_3,
\lambda_3]$ is an $(i-1)$-run of $\A'$,
\item[{\em c)}]
some transition $(i,P_2,V_2,R_3)\to (R_4,\ttLift)\in \mu'$ is
applied to obtain the transition relation $[i, R_3, \lambda_3]
\vdash [i,R_4,\lambda_4]$,
\item[{\em d)}]
some transition $(i,P_3,V_3,R_4)\to (R_5,\ttRight)\in \mu'$ is
applied to obtain the transition relation
$[i,R_4,\lambda_4]\vdash[i,R_5,\lambda_5]$.
\end{itemize}
Thus, to prove part~$(1)$, it suffices to prove the following four
equations. {\footnotesize
\begin{eqnarray}
\label{eq: a1} \bigcup_{p_1\in R_1} \{\gamma : [i,p_1,\lambda_1]
\vdash_{\sA} \gamma \} & = & \{ [i-1,p_2,\lambda_2] : p_2 \in R_2\}
\\
\label{eq: a2} \bigcup_{p_2\in R_2} \{\gamma :
[i-1,p_2,\lambda_2]\vdash_{\sA}^\ast \gamma \textrm{ is an
}(i-1)\textrm{-run}\} & = & \{[i-1,p_3,\lambda_3]: p_3 \in R_3\}
\\
\label{eq: a3} \bigcup_{p_3\in R_3} \{\gamma : [i-1,p_3,\lambda_3]
\vdash_{\sA} \gamma \} & = & \{ [i,p_4,\lambda_4] : p_4 \in R_4\}
\\
\label{eq: a4} \bigcup_{p_4\in R_4} \{\gamma : [i,p_4,\lambda_4]
\vdash_{\sA} \gamma \} & = & \{ [i,p_5,\lambda_5] : p_5 \in R_5\}
\end{eqnarray}
} \vspace{0.5 em}
\par \noindent
{\em Proof of} (\ref{eq: a1}): By item~{\em (a)} above, there is a
transition $(i,P_1,V_1,R_1)\to (R_2,\ttPlace)\in \mu'$. By the
construction of $\mu'$ that $R_2 = E_{i,P_1,V_1}(R_1)$, we
immediately have Equation~\ref{eq: a1}.

\vspace{0.5 em}
\par \noindent
{\em Proof of} (\ref{eq: a2}): By item~{\em (b)} above,
$[i-1,R_2,\lambda_2]\vdash^\ast [i-1, R_3, \lambda_3]$ is an
$(i-1)$-run of $\A'$. Equation~\ref{eq: a2} follows from the
induction hypothesis.

\vspace{0.5 em}
\par \noindent
{\em Proof of} (\ref{eq: a3}): By item~{\em (c)} above, there is a
transition $(i,P_2,V_2,R_3)\to (R_4,\ttLift)\in \mu'$. By the
construction of $\mu'$ that $R_4 = E_{i,P_2,V_2}(R_3)$, we
immediately have Equation~\ref{eq: a3}.

\vspace{0.5 em}
\par \noindent
{\em Proof of} (\ref{eq: a4}): By item~{\em (d)} above, there is a
transition $(i,P_3,V_3,R_4)\to (R_5,\ttRight)\in \mu'$. By the
construction of $\mu'$ that $R_5 = E_{i,P_3,V_3}(R_4)$, we
immediately have Equation~\ref{eq: a4}.

\vspace{1.0 em}

The proof of part~$(2)$ of our claim is very similar to the one of
part~$(1)$. For completeness, we present it here. Let $q_1,q_2 \in
Q$ and
$$
[i,q_1,\theta_1]\vdash^\ast [i,q_2,\theta_2]
$$
be an $i$-run of $\A$.

By our normalization of $\A$, an $i$-run of $\A'$ is a repeated
sequence of transitions relations of the form:
$$
[i,p_1,\lambda_1]\vdash [i-1,p_2,\lambda_2]\vdash^\ast [i-1, p_3,
\lambda_3] \vdash [i,p_4,\lambda_4]\vdash [i,p_5,\lambda_5],
$$
where
\begin{itemize}
\item[{\em a)}]
some transition $(i,P_1,V_1,p_1)\to (p_2,\ttPlace)\in \mu$ is
applied to obtain the transition relation $[i,p_1,\lambda_1]\vdash
[i-1,p_2,\lambda_2]$,
\item[{\em b)}]
the transition relation $[i-1,p_2,\lambda_2]\vdash^\ast [i-1, p_3,
\lambda_3]$ is an $(i-1)$-run of $\A'$,
\item[{\em c)}]
some transition $(i,P_2,V_2,p_3)\to (p_4,\ttLift)\in \mu'$ is
applied to obtain the transition relation $[i-1, p_3, \lambda_3]
\vdash [i-1,p_4,\lambda_4]$,
\item[{\em d)}]
some transition $(i,P_3,V_3,p_4)\to (p_5,\ttRight)\in \mu'$ is
applied to obtain the transition relation
$[i,p_4,\lambda_4]\vdash[i,p_5,\lambda_5]$.
\end{itemize}
The proof of part~$(2)$ follows from the following.
\begin{itemize}
\item
For all $R_1 \subseteq Q$ such that $p_1 \in R_1$, there exists $R_2
\subseteq Q$ such that $p_2 \in R_2$ and
$[i,R_1,\lambda_1]\vdash_{\sA'} [i-1,R_2,\lambda_2]$.
\\
This immediately follows from the fact that $p_2 \in
E_{i,P_1,V_1}(R_1)$ and $(i,P_1,V_1,R_1)\to
(E_{i,P_1,V_1}(R_1),\ttPlace)\in \mu'$.
\item
For all $R_2 \subseteq Q$ such that $p_2 \in R_2$, there exists $R_3
\subseteq Q$ such that $p_3 \in R_3$ and
$[i-1,R_2,\lambda_2]\vdash_{\sA'}^{\ast} [i-1,R_3,\lambda_3]$ is an
$(i-1)$-run.
\\
This follows from the induction hypothesis.
\item
For all $R_3 \subseteq Q$ such that $p_3 \in R_3$, there exists $R_4
\subseteq Q$ such that $p_4 \in R_4$ and
$[i-1,R_3,\lambda_3]\vdash_{\sA'} [i,R_4,\lambda_4]$.
\\
This immediately follows from the fact that $p_3 \in
E_{i,P_2,V_2}(R_3)$ and $(i,P_2,V_2,R_3)\to
(E_{i,P_2,V_2}(R_3),\ttLift)\in \mu'$.
\item
For all $R_4 \subseteq Q$ such that $p_4 \in R_4$, there exists $R_5
\subseteq Q$ such that $p_5 \in R_5$ and
$[i,R_4,\lambda_4]\vdash_{\sA'} [i,R_5,\lambda_5]$.
\\
This immediately follows from the fact that $p_4 \in
E_{i,P_3,V_3}(R_4)$ and $(i,P_3,V_3,R_4)\to
(E_{i,P_3,V_3}(R_4),\ttRight)\in \mu'$.
\end{itemize}
\end{proof}

\subsection{From alternating to nondeterministic}
\label{app: ss: alternating to nondet}

Let $\A = \langle \Sigma, Q,q_0,\mu,F,U\rangle$
be one-way alternating weak $k$-PA.
Adding some extra states,
we can normalize $\A$ as follows.
\begin{itemize}
\item
For every $p \in U$,
if $(i,\sigma,V,p)\to (q,\ttAct)\in \mu$,
then $\ttAct = \ttStay$.
\item
Every pebble can be lifted only
after it reads the right-end marker $\triangleright$.
\item
Only pebble~$k$ can enter a final state
and it does so only after it reads the right-end marker $\triangleright$.
\end{itemize}

We assume that $Q$ is partitioned into $Q_1 \cup \cdots \cup Q_k$
where $Q_i$ is the set of states,
where
pebble $i$ is the head pebble, for each $i=1,\ldots,k$.
We can further partition each $Q_i$ into four sets of states:
$Q_{i,\sttStay}$, $Q_{i,\sttRight}$, $Q_{i,\sttPlace}$, $Q_{i,\sttLift}$ such that
for every $i$, $\sigma$, $V$, $q$ and $p$,
\begin{itemize}
\item
if $q \in Q_{i,\sttStay}$ and $(i,\sigma,V,q) \to (p,\ttAct) \in \mu$,
then $\ttAct = \ttStay$;
\item
if $q \in Q_{i,\sttRight}$ and $(i,\sigma,V,q) \to (p,\ttAct) \in \mu$,
then $\ttAct = \ttRight$;
\item
if $q \in Q_{i,\sttPlace}$ and $(i,\sigma,V,q) \to (p,\ttAct) \in \mu$,
then $\ttAct = \ttPlace$;
\item
if $q \in Q_{i,\sttLift}$ and $(i,\sigma,V,q) \to (p,\ttAct) \in \mu$,
then $\ttAct = \ttLift$.
\end{itemize}
Now we define a nondeterministic weak $k$-PA
$\A' = \langle \Sigma,Q',q_0',\mu',F'\rangle$, where
\begin{itemize}
\item
$Q'$ consists of states of the form $(S_k,S_{k-1}\ldots,S_{k-j})\in 2^{Q_{k}}\times 2^{Q_{k-1}}\times 2^{Q_{k-j}}$,
where $0\leq j \leq k-1$;
\item
$q_0'=\{q_0\}$;
\item
$F' = 2^{F} - \{\emptyset\}$.
\end{itemize}

The set $\mu'$ contains the following transitions.
For every $i=1,2,\ldots,k$,
for every $V \subseteq \{i+1,\ldots,k\}$,
for every $(S_k,S_{k-1},\ldots,S_i) \in Q'$,
for every $\sigma \in \Sigma$,
we have the following transitions.
\begin{itemize}
\item
If $S_i$ contains a state $q \in U$,
then
$$
(i,\sigma,V,(S_k,S_{k-1},\ldots,S_i)) \to ((S_k,S_{k-1},\ldots,(S_i - \{q\})\cup U_q),\ttStay) \in \mu'
$$
where
$U_q = \{p \mid (i,\sigma,V,q)\to(p,\ttStay) \in \mu\}$.

\item
If $S_i$ contains a state $q \in Q_{i,\sttStay}$ and $S \cap U = \emptyset$,
then
$$
(i,\sigma,V,(S_k,S_{k-1},\ldots,S_i)) \to ((S_k,S_{k-1},\ldots,(S_i-\{q\})\cup N_q),\ttStay) \in \mu'
$$
where
$N_q \subseteq \{p \mid (i,\sigma,V,q)\to(p,\ttStay) \in \mu\}$.

\item
If $S_i$ contains a state $q \in Q_{i,\sttPlace}$ and $S_i \cap Q_{i,\sttStay} = \emptyset$,
then
$$
(i,\sigma,V,(S_k,S_{k-1},\ldots,S_i)) \to ((S_k,S_{k-1},\ldots,S_i-\{q\},\{p\}),\ttPlace) \in \mu'
$$
where $(i,\sigma,V,q)\to(p,\ttPlace) \in \mu$.

\item
If $S_i \subseteq Q_{i,\sttRight}$,
then
$$
(i,\sigma,V,(S_k,S_{k-1},\ldots,S_i)) \to ((S_k,S_{k-1},\ldots,S_i') ,\ttRight) \in \mu'
$$
where $S_i' = \{p \mid (i,\sigma,V,q)\to(p,\ttRight) \in \mu \textrm{ and }q \in S\cap Q_i\}$.

\item
If $S_i \subseteq Q_{i,\sttLift}$,
then
$$
(i,\triangleright,V,(S_k,S_{k-1},\ldots,S_{i+1},S_i)) \to ((S_k,S_{k-1},\ldots,S_{i+1}\cup R) ,\ttLift) \in \mu'
$$
where $R = \{p \mid (i,\sigma,V,q)\to(p,\ttLift) \in \mu \textrm{ and }q \in S_i\}$.
\end{itemize}
The proof that $L(\A)=L(\A')$ is pretty much similar to the one in
the previous subsection.

Let $S_k,S_{k-1},\ldots,S_i$ and $\theta$ such that
$$
[k,\{q_0\},\theta_0] \vdash_{\sA'}^* [i,(S_k,S_{k-1},\ldots,S_i),\theta].
$$
For each $j=k,k-1,\ldots,i$, we define $\theta_j$ to be
$\theta$ restricted to the domain $\{k,k-1,\ldots,j\}$.
Then, by straightforward induction on $i$,
we can show that
$$
[k,q_0,\theta_0] \vdash_{\sA}^* [j,q,\theta_j], \textrm{ for each }j=k,k-1,\ldots,i.
$$

Now we show the converse direction.
Let $\theta$ be an assignment for pebbles~$k,k-1,\ldots,i$ and
we denote by $\theta_j$ the pebble assignment $\theta$ restricted to the domain $\{k,k-1,\ldots,j\}$
when $j\geq i$.
Let $S_k,S_{k-1},\ldots,S_i$ be subsets of $Q_k,Q_{k-1},\ldots,Q_i$, respectively, and
let $p_k,p_{k-1},p_{i+1}$ be an element of $Q_k,Q_{k-1},\ldots,Q_i$ such that
for each $j=k,k-1,\ldots,i$,
\begin{itemize}
\item
$p_j \notin S_j$;
\item
$[k,q_0,\theta_0] \vdash_{\sA}^* [j,q,\theta_j]$, for each $q \in S_j$.
\end{itemize}
Then, by straightforward induction on $i$,
we can show that
$$
[k,\{q_0\},\theta_0] \vdash_{\sA'}^* [i,(S_k,S_{k-1},\ldots,S_i),\theta].
$$

\end{document}